%% file: main.tex
\documentclass[a4paper]{article}

\usepackage[english]{babel}
\usepackage[utf8x]{inputenc}
\usepackage[T1]{fontenc}

\usepackage{framed}

\usepackage{color}

\usepackage{float}

\renewenvironment{shaded}{%
  \MakeFramed{\advance\hsize-\width \FrameRestore\FrameRestore}}%
 {\endMakeFramed}
\definecolor{mygray}{gray}{0.9}

\usepackage[a4paper,top=3cm,bottom=2cm,left=3cm,right=3cm,marginparwidth=1.75cm]{geometry}

\usepackage{amsmath}
\usepackage{amssymb}
\usepackage{amsthm}
\usepackage{proof}
\usepackage{graphicx}
\usepackage{subcaption}
\usepackage[colorinlistoftodos]{todonotes}
\usepackage[colorlinks=true, allcolors=blue]{hyperref}
\usepackage{tikz-cd}
\usepackage{wrapfig,lipsum}
\usepackage{authblk}

\newcounter{Counter}
\newtheorem{Theorem}[Counter]{Theorem}

\newtheorem{Lemma}[Counter]{Lemma}
\newtheorem{Corollary}[Counter]{Corollary}

\newcommand{\toffoli}{\mathrm{TOFFOLI}}
\newcommand{\cnot}{\mathrm{CNOT}}
\newcommand{\nnot}{\mathrm{NOT}}

\title{Explicit lower bounds on strong quantum simulation}
\author[1,2]{ Cupjin Huang}
\author[1,3]{Michael Newman}
\author[1]{Mario Szegedy}

\affil[1]{\small \textit{Aliyun Quantum Laboratory, Alibaba Group, Bellevue, WA 98004, USA}}
\affil[2]{\small \textit{Department of Computer Science, University of Michigan, Ann Arbor, MI 48109, USA}}
\affil[3]{\small \textit{Department of Mathematics, University of Michigan, Ann Arbor, MI 48109, USA}}
\date{}
\begin{document}
\maketitle
%## Abstract
\input{abstract}
%* Simulators are not clever; we give *explicit* lower bounds to these simulation techniques (so that our paper is not trivial)
%* Monotone simulators have *explicit* unconditional exponential time lower bounds
%* Amplitude-wise simulators have *explicit* plausible exponential time lower bounds

%## Motivation
\input{motivation}
%* Why is circuit simulation so important?
% quantum supremacy, testing quantum devices, growing literature
%* ***Define*** amplitude-wise simulation vs. an example of generating a distribution. Highlight the difference between "strong" and "weak" simulations, and more "fine-grained"
% Our contribution: nobody expects ss to be feasible, however no one has made it accurate

%* Many prominent simulation techniques are cut from the same cloth
%* *Explicit* lower bounds are important (reductions matter) though it deals with algorithmic cleverness rather than memory/parallelization cleverness
%* *Unconditional* lower bounds are precious 

%## Results
%* Roughly define monotone simulation
%* Summarize unconditional explicit lower bounds (mention permanent reduction proof technique)
%* General amplitude-wise simulation has conditional explicit lower bounds (mention #SAT reduction proof technique)

%---------------

%## "Boring" simulation is unconditionally "boring"
%### Motivation of the model
\input{preliminaries}

\input{unconditional_model}

%* Define several methods for amplitude-wise simulation and prove that they are monotone
%* Address other unconditional lower bound paper (show separation in appendix)

%### Proof of the unconditional comlexity
\input{unconditional_proof}

%* Define a quantum circuit tensor network
%* Embed permanent into QCTN
%* Use existing monotone arithmetic lower bounds to give explicit unconditional lower bounds

%## "Exciting" simulation is conditionally "boring"
\input{conditional}

%### Preliminaries
%* Define approximate amplitude-wise simulation
%* Introduce SETH
%* Introduce reversible computing (Poly-time, sublinear-overhead reduction of exptime task; constants matter)

%### Proof fo the conditional complexity
%* Construct reversible circuit for #SAT
%* Apply approximate #SAT
%* Conclude approximability vs. running time tradeoff by reducing GAP-SAT to approximate #SAT

%### Discussion of the result
%* Potential hardness of approxiamte #SAT for better lower bound

%----------

%## Conclusion (Ideas)
\input{conclusion}
%* Explicit lower bounds for existing simulation techniques
%* Thoughts on completely different approaches
%* Runtime estimates of current state-of-the-art restricted circuit classes?
%* Space vs. time considerations?

%-------------
\input{openproblems}

\bibliographystyle{alpha}
\bibliography{biblio}

\end{document}

%% file: abstract.tex
\begin{abstract}
We consider the problem of strong (amplitude-wise) simulation of $n$-qubit quantum circuits, and identify a subclass of simulators we call monotone.  This subclass encompasses almost all prominent simulation techniques.  We prove an {\em unconditional} (i.e.\ without relying on any complexity theoretic assumptions) and \emph{explicit} $(n-2)(2^{n-3}-1)$ lower bound on the running time of simulators within this subclass.  Assuming the Strong Exponential Time Hypothesis (SETH), we further remark that a universal simulator computing \emph{any} amplitude to precision $2^{-n}/2$ must take at least $2^{n - o(n)}$ time.  Finally, we compare strong simulators to existing SAT solvers, and identify the time-complexity below which a strong simulator would improve on state-of-the-art SAT solving.
\end{abstract}

%% file: motivation.tex
\section{Introduction}

As quantum computers grow in size, so too does the interest in simulating them classically.  Intimately tied to our ability to simulate these circuits is the notion of \emph{quantum supremacy}, which tries to quantify the threshold beyond which quantum devices will outperform their classical counterparts \cite{preskill2012quantum,boixo2016characterizing}.  Furthermore, intermediate-size quantum devices need to be tested, a task that can be assisted by efficient simulation \cite{preskill2018quantum}.    

Consequently, there has been a growing body of literature aimed at improving, testing, and classifying existing simulation techniques \cite{boixo2016characterizing,boixo2017simulation,aaronson2017complexity,chen2018classical,chen2018,haner20170,pednault2017breaking,bravyi2016improved,nest2008classical}.  One major division within the landscape of quantum simulators is between strong and weak simulation.  Strong quantum simulators compute the amplitude of a particular outcome, whereas weak simulators only sample from the output distribution of a quantum circuit.

In this work we focus exclusively on strong simulation, and give compelling evidence that it is fundamentally unscalable.  This is hardly surprising: it is well-known that perfectly accurate strong simulation is $\#P$-hard.  However, we give \emph{explicit}, and in some contexts \emph{unconditional} evidence for the hardness of strong simulation.

We identify a large class of simulation techniques we call \emph{monotone}.  This class includes most of the known techniques for general quantum circuit simulation.  We place unconditional lower bounds on simulators within this class.  In particular, we show that there exists a simple quantum circuit which will take any such simulator at least $(n-2)(2^{n-3}-1)$ time to simulate.

Next, we consider general strong quantum simulators which allow for approximation.  Assuming the Strong Exponential Time Hypothesis, we remark that there exist simple quantum circuits for which any strong simulator with accuracy $2^{n}/2$ must take at least $2^{n - o(n)}$ time to simulate.  Finally, we compare strong simulators to state-of-the-art SAT solvers.  We identify parameter thresholds on strong simulators above which such a simulator would imply immediate gains on existing solvers.

While these bounds are concerned with algorithmic complexity rather than clever memory allocation and parallelization, they indicate that strong simulation of hundreds of qubits is fundamentally infeasible.  However, not all general circuit simulators are strong, nor are they all monotone.  An excellent counterexample is a pair of non-monotone simulators constructed by Bravyi and Gosset~\cite{bravyi2016improved} to simulate quantum circuits of low $T$-complexity $t$.  One is a strong simulator with a relevant exponential factor of $2^{0.47t}$, and the other is a weak simulator with an exponential factor of $2^{0.23t}$. 

While these particular simulators will not offer an advantage in simulating general quantum circuits, they illustrate the advantage of weak simulation. Our evidence for the fundamental hardness of strong simulation indicates that intrinsically different \emph{weak} simulators \emph{must} be the focus in order to scale up.

%% file: preliminaries.tex
\section{Preliminaries}
\noindent {\bf Tensor Networks.} A \emph{tensor} is simply a multidimensional array. The dimension of the array is called the \emph{rank} of the tensor. The simplest examples of tensors are vectors and matrices, which are rank 1 and rank 2 tensors, respectively. Tensors are usually expressed graphically, see Figure~\ref{tensors}. For simplicity, we assume that each index ranges in $\{0,1\}$.

\begin{figure}[h!]
\centering
\begin{subfigure}[tc]{0.3\textwidth}
\centering
\includegraphics[width = \textwidth]{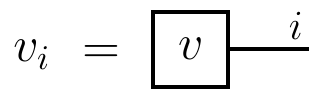}
\caption{A vector is a $1$-tensor.}
\end{subfigure}
~,
\begin{subfigure}[tc]{0.3\textwidth}
\centering
\includegraphics[width = \textwidth]{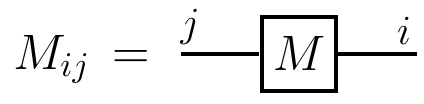}
\caption{A matrix is a $2$-tensor.}
\end{subfigure}
~,
\begin{subfigure}[tc]{0.3\textwidth}
\centering
\includegraphics[width = \textwidth]{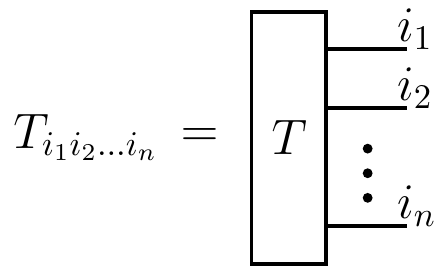}
\caption{An $n$-tensor.}
\end{subfigure}
\caption{Examples of tensors.}
\label{tensors}
\end{figure}

A \emph{tensor network} is a collection of tensors together with identifications among the open indices of the tensors. Throughout the paper, we restrict ourselves to \emph{closed} tensor networks. These are tensor networks representing a scalar obtained by summing over all identified open indices. See Figure~\ref{tensornet} for an example. The \emph{shape} of a tensor network is the information given by the ranks of each tensor, together with the identifications of the open indices.  In particular, it is the information which tells you how to contract the tensor network, but not what the values of the tensors are.

\begin{figure}[h!]
\centering
\includegraphics[width = 0.7\textwidth]{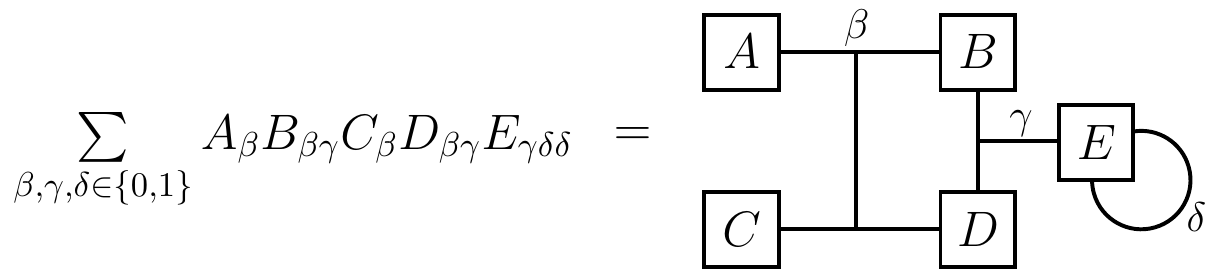}
\caption{Example of a (closed) tensor network.}
\label{tensornet}
\end{figure}

%\noindent{\bf Circuit.} A \emph{circuit} $C$ is a directed acyclic graph $G=(V,E)$ together with a set of source nodes $S \subseteq V$ called inputs and a set of sink nodes $T \subseteq V$ called outputs. Each node in $V\setminus S$ represents a \emph{gate}: it is associated with some local computation. The \emph{size} of $C$ is defined as $s(C):=|V|$, and the \emph{depth} $d(C)$ is the size of the longest directed path in graph $G$. The following two categories of circuits are considered in this paper.

\medskip

\noindent{\bf Monotone Arithmetic Circuits.} In a monotone arithmetic circuit, each gate has fan-in degree $2$ and is either a $+$-gate or a $\times$-gate. Furthermore, we assume that there is a single output node, see Figure~\ref{mac} for an example.

\begin{figure}[h!]
\centering
\includegraphics[width = 0.4\textwidth]{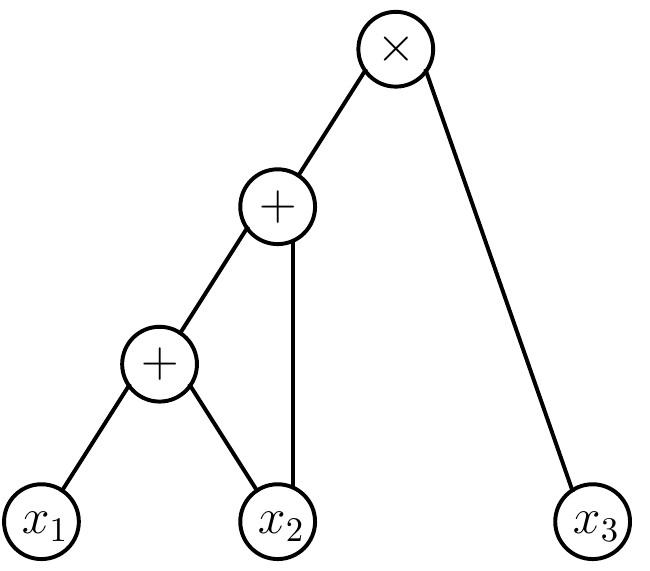}
\caption{Example of a monotone arithmetic circuit computing the polynomial $(x_1+2x_2)x_3$.}
\label{mac}
\end{figure}

\noindent{\bf Quantum Circuits.} In a \emph{quantum circuit} $C$, each gate has the same in- and out-degree, and represents a unitary matrix. The width of a quantum circuit $w(C)$, is the number of input nodes to the circuit.  A special case of a quantum circuit is a \emph{classical reversible} circuit.

\medskip

\noindent{\bf Strong and Weak Simulation.}  A \emph{strong simulation} of the quantum circuit $C$ computes the value $\langle 0|C|x\rangle$ for some specified $x$, and a \emph{weak simulation} of the circuit samples from the distribution $p(x)=|\langle 0|C|x\rangle|^2$.

%% file: unconditional_model.tex
\section{An unconditional lower bound for monotone methods}

When proving a lower bound, the model is at least as important as the bound itself.  If the model is too restrictive, then the lower bound loses its worth.  In this section, we introduce the \emph{monotone method}, which describes a strong quantum simulation methodology.  Although the model is restrictive enough to permit \emph{unconditional} lower bounds, it also encompasses the majority of existing strong simulation techniques.

Before we define a monotone method rigorously, we describe it informally as a game. The game is played between a referee and you, and the aim of the game is to simulate a quantum circuit. First, the referee hands you a picture of a quantum circuit, but with some information missing.  He has erased all of the nonzero coefficients in the gates appearing, replacing each with a different variable. You are allowed to spend as long as you like preparing your strategy; you may even use infinite time to do so.  

When you are finally ready, you must commit to a \emph{monotone} arithmetic circuit, with inputs given by the variables.  Your monotone arithmetic circuit must simulate the original circuit with perfect accuracy \emph{no matter what values the variables take}, and the time-complexity of your task is measured only by the length of your arithmetic circuit. With this game in mind, we set out to define a monotone method.

\subsection{The skeleton of a tensor network}
In the game we describe, the partial information about the circuit is very specific.  Informally, we call the picture that the referee gives us the \emph{skeleton} of the quantum circuit.  More generally, we can consider the quantum circuit as a closed tensor network, and define an associated skeleton.

\medskip

\noindent{\bf Skeleton of a Tensor Network.} The \emph{skeleton} of a tensor network is the hypergraph associated to the tensor network along with the locations of the nonzero entries in each tensor.  Furthermore, we call the skeleton \emph{closed} if the hypergraph is closed.

\medskip

We further define the (closed) skeleton of a quantum circuit as the skeleton of its associated tensor network. To any \emph{closed} skeleton, we can define an associated polynomial.

\medskip

\noindent{\bf Associated Polynomial.} Given a closed skeleton $S$, we can associate to it a polynomial \emph{$p(S)$} in the following way. First, replace each nonzero entry in each tensor appearing with a different variable.  We can then regard $S$ as a closed tensor network.  Define $p(S)$ to be the polynomial obtained from contracting $S$. 

\medskip

See Figure~\ref{skeleton} for an example of a quantum circuit, its skeleton, and its associated polynomial.  Note that $p(S)$ is independent of the sequence of contractions and has nonnegative coefficients.  To any quantum circuit $C$, we can also define the polynomial $p(C)$ to be the polynomial associated to its skeleton.

\begin{figure}[h]
\centering

\begin{subfigure}[tc]{0.48\textwidth}
\centering
\includegraphics[width = \textwidth]{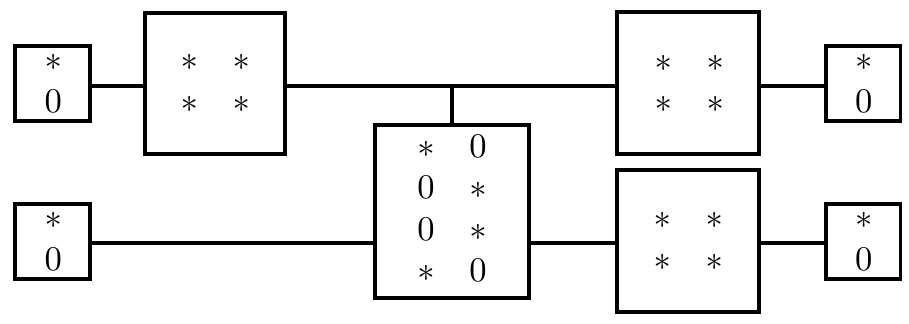}
\end{subfigure}
,\hspace{0.1cm}
\begin{subfigure}[tc]{0.48\textwidth}
\centering
\includegraphics[width = \textwidth]{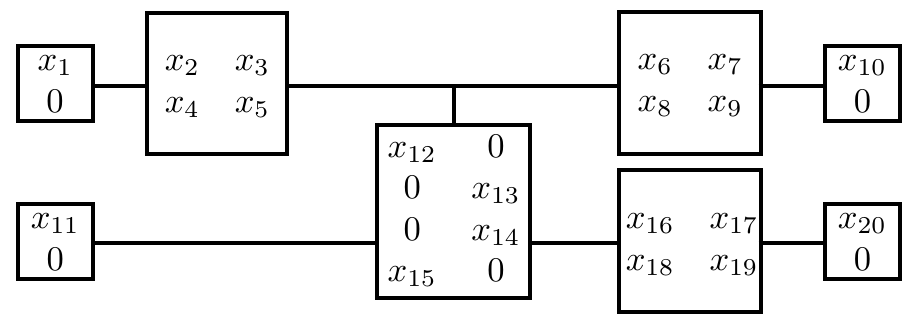}
\end{subfigure}
\caption{The left-hand diagram represents the skeleton associated to the quantum circuit $\langle 0 | H_1H_2\otimes CX_{1\rightarrow2}\otimes H_1\otimes I_2|0\rangle$, where the subscripts indicate the wires on which the gates act.  According to the variables introduced in the right-hand diagram, its associated polynomial is $p(x_1, \ldots, x_{20}) = x_1x_{10}x_{11}x_{20}(x_2x_6x_{12}x_{16} + x_3x_8x_{14}x_{18})$.}
\label{skeleton}
\end{figure}

\subsection{Monotone methods}

We will now define a monotone method.  For any arithmetic circuit $A$, let $p(A)$ be its associated polynomial.

\medskip

\noindent{\bf Monotone Method.} A \emph{monotone method} $M$ is a mapping from closed skeletons to \emph{monotone} arithmetic circuits that preserves the associated polynomials: for all closed skeletons $S$, $p(M(S)) = p(S)$.  We define the \emph{monotone complexity} of the skeleton $S$ under the mapping $M$ as $|M(S)|$. We can further extend a monotone method to a mapping on quantum circuits through their associated skeletons.  

\medskip

In particular, $M$ itself can take exponential time to compute or even be non-uniform, allowing for uncomputable strategies.  The complexity is measured only in terms of the complexity of the resulting arithmetic circuit. See Figure~\ref{high_level} for a high-level description relating strong quantum simulators to monotone methods.

\begin{figure}[htb!]
\begin{center}
\includegraphics[width = 0.9 \textwidth]{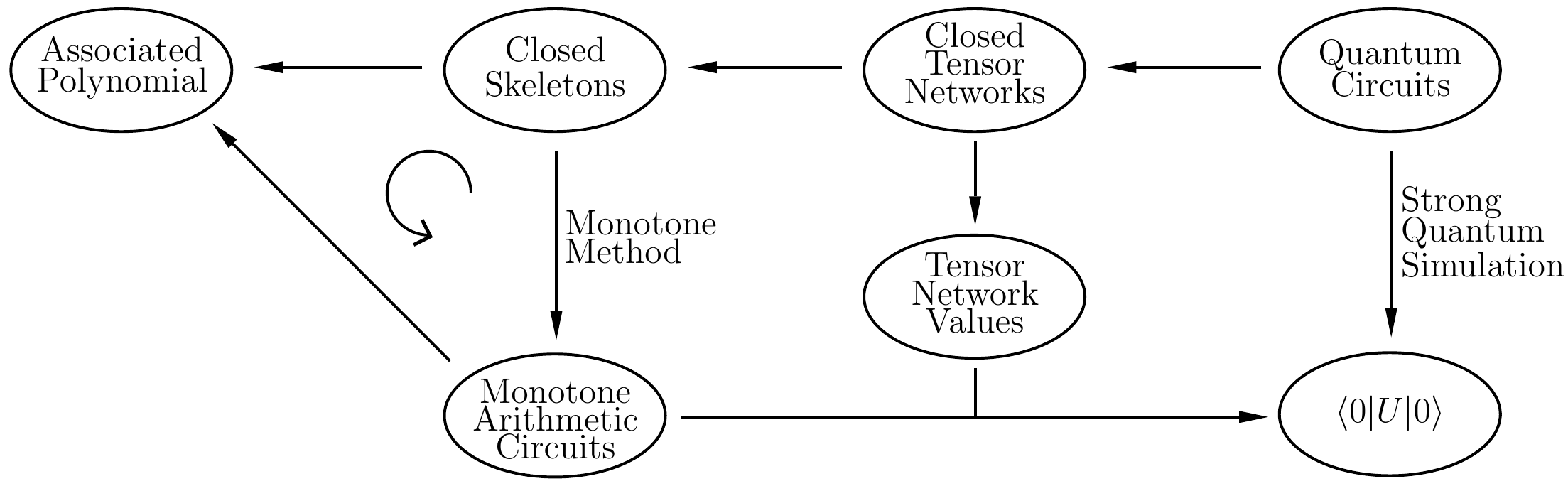}
\end{center}
\caption{The relation between strong quantum simulation and monotone methods.  A monotone method is any map from closed skeletons to monotone arithmetic circuits that makes the diagram commute.}
\label{high_level}
\end{figure}

To demonstrate the applicability of this model, we give a list of prominent strong simulation techniques, and show that they are all monotone methods.

\begin{enumerate}

\item Feynman's path integral~\cite{feynman1965quantum} is the simplest example of a monotone method applied to general tensor network contraction, involving only additions and multiplications of \emph{all} coefficients appearing in the circuit.

\item Markov and Shi~\cite{markov2008simulating} propose a tensor network contraction algorithm consisting of two phases. In the first phase, a nearly optimal ordering of contractions is decided. In the second phase, the tensor network is contracted accordingly. Given a general tensor network, finding the optimal ordering of contractions is an NP-hard problem; although there may be a contraction ordering with low complexity, finding it may take exponential time. One could carefully choose the tradeoff between these two phases so that the total complexity is split evenly. 

In the second phase, contracting the tensors one by one is realized by a monotone arithmetic circuit. In particular, the algorithm itself is a monotone method.  The complexity of \emph{the second phase alone} is counted in the complexity of our model.

\item Several works~\cite{boixo2017simulation, pednault2017breaking, chen2018classical, chen2018}  preprocess the quantum circuit in a way that simplifies the tensor network, and ultimately results in a lower complexity during contraction. Specifically, they reduce the complexity of contracting diagonal gates. This preprocessing is oblivious to the nonzero coefficients of the gates.  After deciding an order of contractions, the resulting procedure is realized by a monotone arithmetic circuit with respect to the non-zero entries.  The preprocessing step is illustrated in Figure~\ref{fig:prep}.

\begin{figure}[h!]
\centering
\begin{subfigure}[tc]{0.4\textwidth}
\centering
\includegraphics[width = \textwidth]{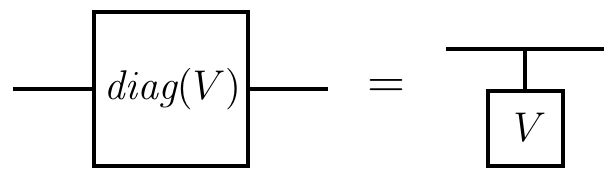}
\caption{Simplification of a one-qubit diagonal gate.}
\end{subfigure}
~
\begin{subfigure}[tc]{0.4\textwidth}
\centering
\includegraphics[width = \textwidth]{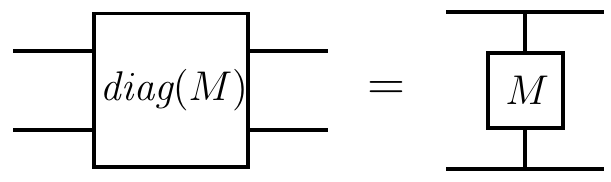}
\caption{Simplification of a two-qubit diagonal gate.}
\end{subfigure}
\\

\vspace{0.2cm}

\begin{subfigure}[tc]{0.8\textwidth}
\centering
\includegraphics[width = \textwidth]{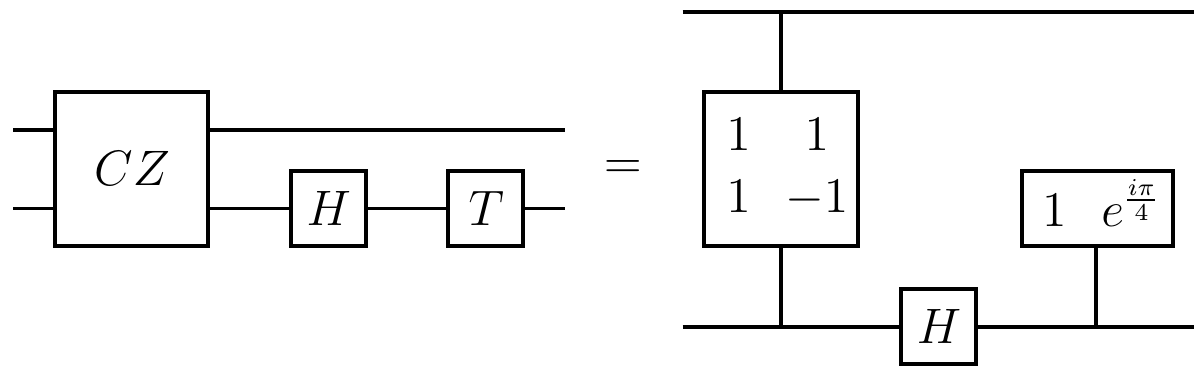}
\caption{Naively contracting the left tensor network takes $32+8=40$ multiplication steps. Since the $CZ$ and $T$ gates are both diagonal, we can simplify the tensor network. The network on the right reduces the number of multiplications to $8+4=12$.}
\end{subfigure}
\caption{Preprocessing a tensor network to reduce the cost of contraction.}
\label{fig:prep}
\end{figure}

\item In~\cite{aaronson2017complexity}, Aaronson and Chen propose a family of algorithms suited to different time-space trade-offs. These algorithms interpolate between naively evaluating the circuit, and implementing a variation of Savitch's theorem. Given a quantum circuit, one can decide the optimal algorithm according to the space and time constraints.  Whichever algorithm you choose, the resulting process can be described by a monotone arithmetic circuit.  
\end{enumerate}

\medskip

There are many more techniques that qualify as monotone methods, such as clever sparse matrix multiplication.  However, there are also many tricks which do \emph{not} fit into the monotone method.  There exist tensor contractions which are non-monotone, first augmenting the network to introduce time-saving cancellations.  Most famous among these is Strassen's fast matrix multiplication algorithm \cite{strassen1969gaussian}.  Recognizing that a circuit belongs to a restricted family such as Clifford or matchgate circuits can give exponential time savings \cite{gottesman1998heisenberg, valiant2002quantum, bravyi2016improved,aaronson2004improved,bravyi2016trading}.  Even recognizing small circuit identities may save on complexity \cite{pednault2017breaking}, although this problem may be hard in general \cite{ji2009non}. Finally, monotone methods are oblivious to the unitarity of the gates, which might be taken advantage of. However, for general quantum circuit simulation, we emphasize that the majority of savings are performed by tricks that fit within the monotone method framework.

%% file: unconditional_proof.tex
\subsection{An unconditional lower bound}

We will now prove an \emph{unconditional} lower bound on the time-complexity of a strong quantum simulator that uses a monotone method.  The proof concept is simple: we will repurpose the simulator to evaluate the permanent with constant overhead.  Using existing unconditional lower bounds on monotone circuit evaluation of the permanent~\cite{jerrum1982some}, we can impose unconditional lower bounds on the time-complexity of the simulator.

\medskip

\noindent {\bf Permanent.} The permanent of an $n \times n$ matrix $M$ is given by $\sum_{\sigma \in S_n}\prod_{i=1}^n M_{i,\sigma(i)}$.  We note that there have been several works relating hard polynomial problems to the hardness of quantum circuits \cite{montanaro2017quantum, bouland2018quantum}, and the permanent is a particularly prominent example \cite{rudolph2009simple, aaronson2011computational}.

\begin{Theorem}[Jerrum \& Snir] \label{Jerrum}
A monotone arithmetic circuit computing the permanent of any $n \times n$ matrix must have size at least $n(2^{n-1}-1)$.
\end{Theorem}

First, we construct a hard family of skeletons, parametrized by $n$.  This family is hard in the sense that a monotone restriction of its associated polynomial computes the permanent of an $n \times n$ matrix.  We then use Theorem~\ref{Jerrum} to lower bound the time-complexity of any monotone method applied to this skeleton.  Finally, we find one (among many) quantum circuits with this skeleton.  We conclude the following.

\begin{Theorem} 
There exists a quantum circuit $C$ of width $n+2$ and depth $3n^2 + 1$ such that for any monotone method $M$, $|M(C)|\geq n(2^{n-1}-1)$.
\label{theo:uncon}
\end{Theorem}

%\begin{Theorem}
%There exists a quantum circuit of width $n + \log(n)$ and depth $n \log(\log(n))$ such that the cleverest monotone method simulating this circuit takes time at least $n(2^{n-1}-1)$.
%\end{Theorem}

We want to re-emphasize that this lower bound is explicit and unconditional.  If God herself were monotone, then this is how long it would take her to simulate these circuits.\footnote{However, we strongly suspect that God is \emph{not} monotone.}

\begin{proof}[Proof of Theorem~\ref{theo:uncon}]
Consider the skeleton $S$ defined in the figure below.

%Figure~\ref{fig:C_perm}. Note that this skeleton corresponds to a circuit of width $n+2$ and depth $3n^2+1$.  

\begin{figure}[H]
\centering
\begin{subfigure}[t]{0.4\textwidth}
\includegraphics[width = \textwidth]{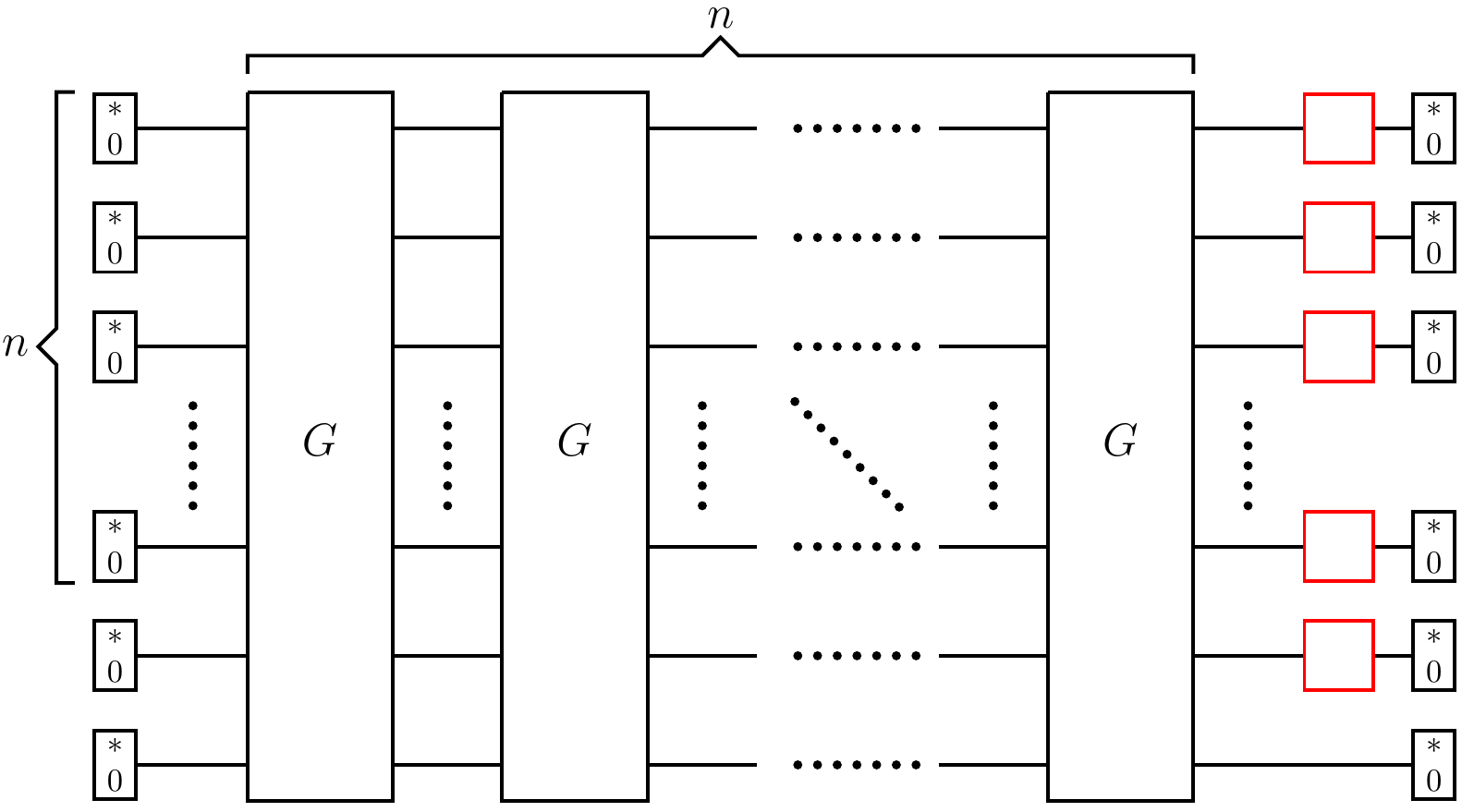}
\caption{Skeleton defined in terms of $G$.}
\label{fig:perm_2}
\end{subfigure}
~
\begin{subfigure}[t]{0.55\textwidth}
\includegraphics[width = \textwidth]{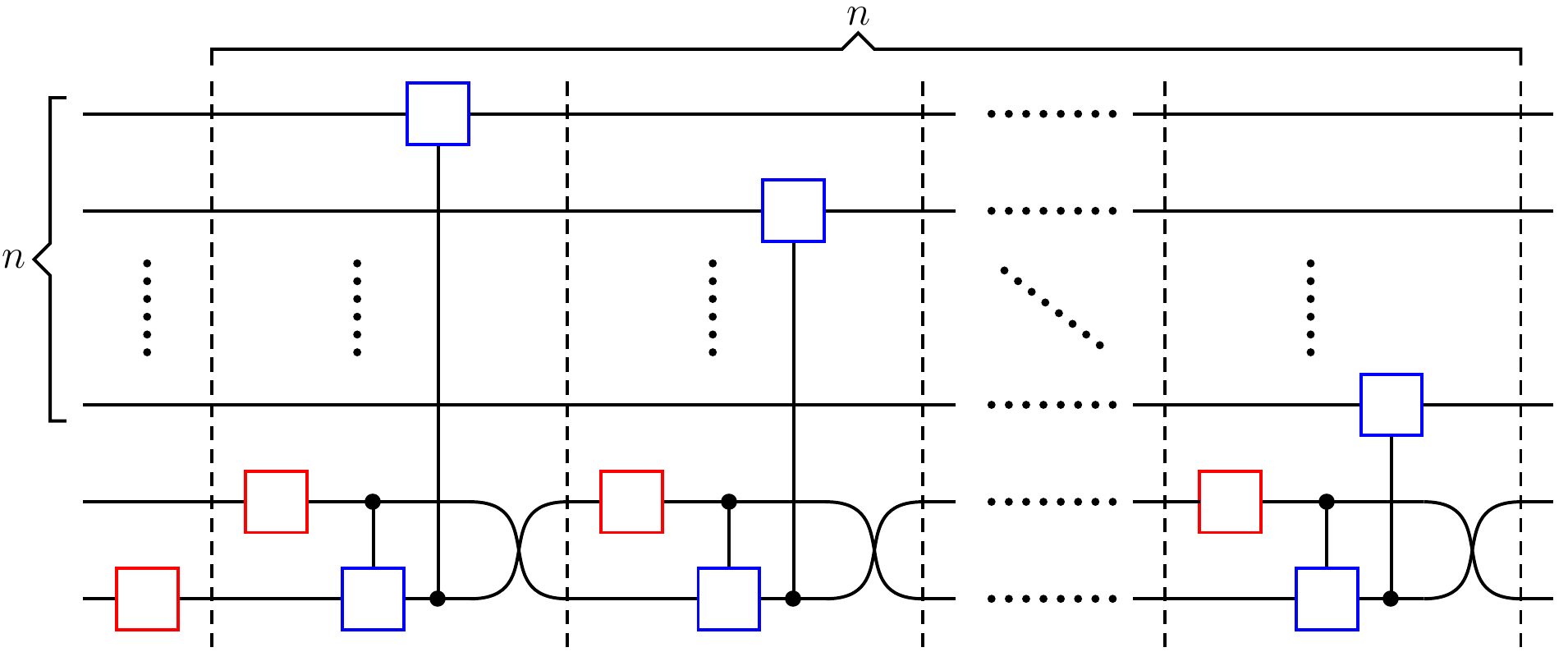}
\caption{Gadget $G$ used in Figure~\ref{fig:perm_2}.}
\label{fig:perm_1}
\end{subfigure} \\ \vspace{0.2cm}
\begin{subfigure}[tc]{0.3\textwidth}
\centering
\includegraphics[width = \textwidth]{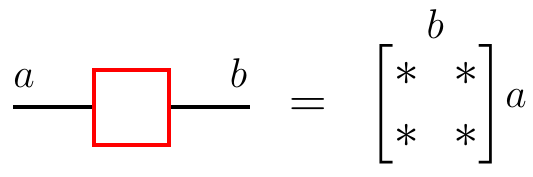}
\caption{
\label{fig:perm_3}}
\end{subfigure}
\hspace{1.0cm}
\begin{subfigure}[tc]{0.3\textwidth}
\centering
\includegraphics[width = \textwidth]{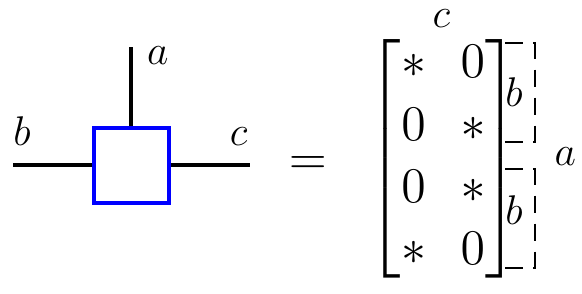}
\caption{
\label{fig:perm_4}}
\end{subfigure}
\caption{The skeleton we use to prove Theorem~\ref{theo:uncon}. The whole circuit is depicted in (a), and the gadget $G$ is depicted in (b).
The nonzero locations of the red component are shown in (c) and the nonzero locations of the blue component are shown in (d).
To obtain a quantum circuit, we can replace the red components with 
Hadamard gates and the blue components with CNOT gates.}
\label{fig:C_perm}
\end{figure}

The polynomial $p(S)$ defined by the skeleton in Figure~\ref{fig:C_perm} is not the permanent. However, we can replace the variables of $p(S)$
with a combination of $x_{ij}$, $0$, and $1$ so that it becomes the permanent of the matrix $(x_{ij})$. Thus, if there were 
a monotone arithmetic circuit computing $p(S)$, then a potentially smaller circuit would compute the permanent. The replacements are shown in Figure~\ref{fig:T_perm}.  The pictorial proof that the resulting circuit computes the permanent can be found in Figure \ref{prooffig}.

\begin{figure}[H]
\centering
\begin{subfigure}[h!]{0.4\textwidth}
\includegraphics[width = \textwidth]{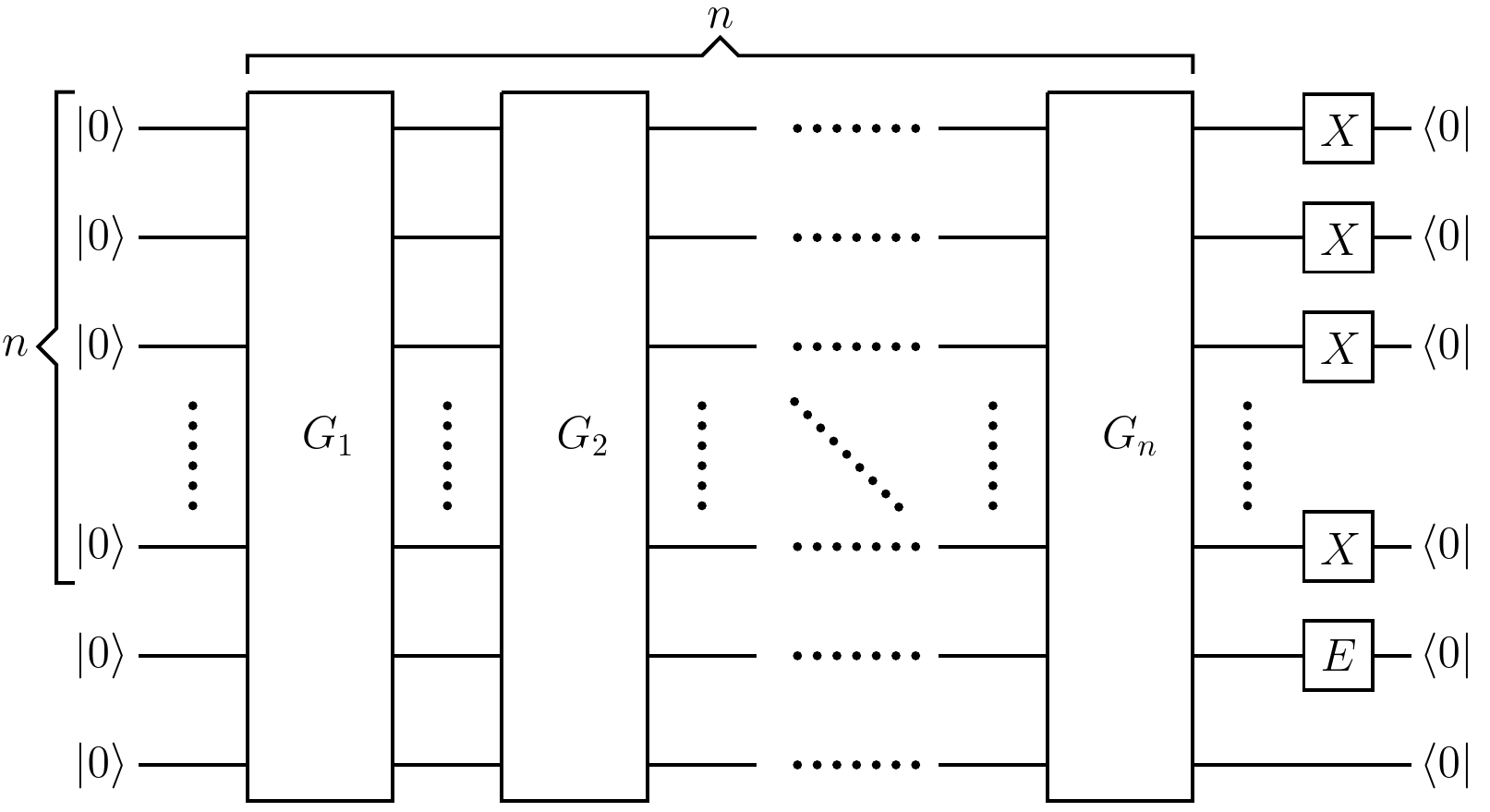}
\caption{Tensor network $T_{perm}$.}
\label{fig:perm_4}
\end{subfigure}
~
\begin{subfigure}[h!]{0.55\textwidth}
\includegraphics[width = \textwidth]{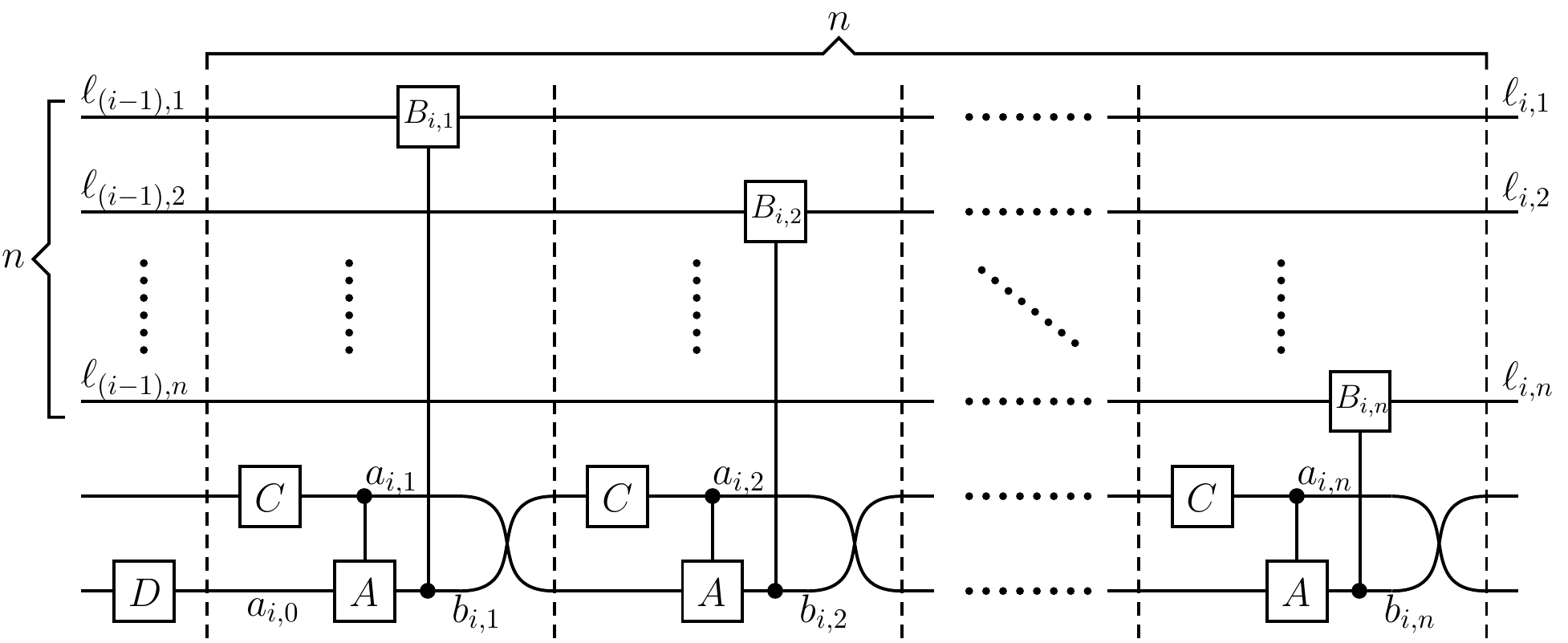}
\caption{Gadget $G_i$ in Figure~\ref{fig:perm_4}.}
\label{fig:perm_3}
\end{subfigure}

\centering
\begin{subfigure}[h!]{0.3\textwidth}
\centering
\includegraphics[width = \textwidth]{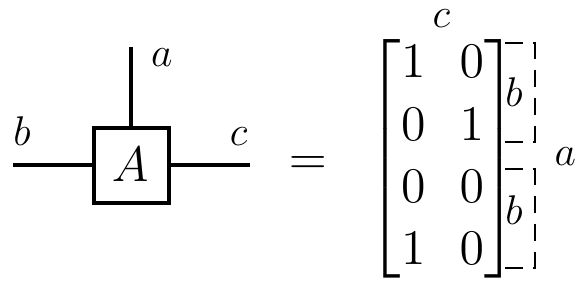}
\end{subfigure}
~,
\begin{subfigure}[h!]{0.3\textwidth}
\centering
\includegraphics[width = \textwidth]{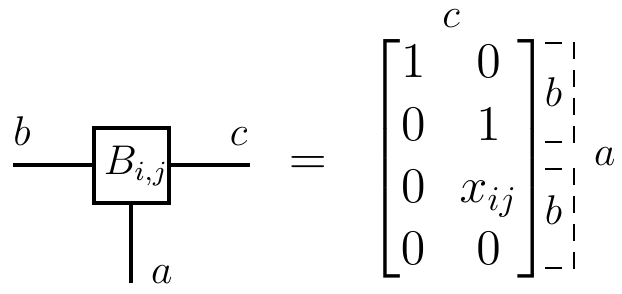}
\end{subfigure}
~,
\begin{subfigure}[h!]{0.3\textwidth}
\centering
\includegraphics[width = \textwidth]{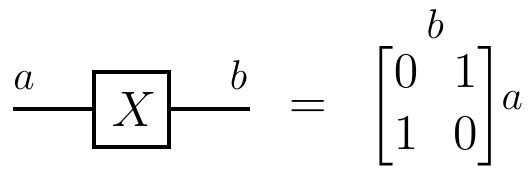}
\end{subfigure}
,\\
\begin{subfigure}[h!]{0.3\textwidth}
\centering
\includegraphics[width = \textwidth]{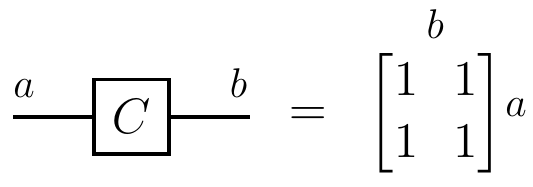}
\end{subfigure}
~,
\begin{subfigure}[h!]{0.3\textwidth}
\centering
\includegraphics[width = \textwidth]{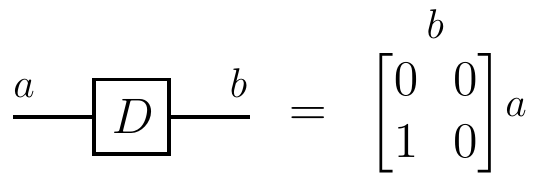}
\end{subfigure}
~,
\begin{subfigure}[h!]{0.3\textwidth}
\centering
\includegraphics[width = \textwidth]{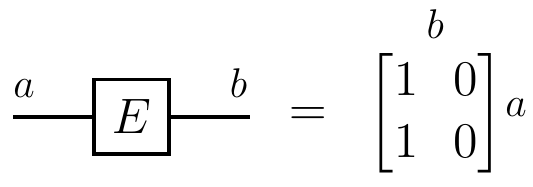}
\end{subfigure}
\caption{The tensor network $T_{perm}$ that realizes the permanent. \label{fig:T_perm}}
\end{figure}

\begin{figure}[H]
\centering
\begin{subfigure}[h!]{0.9\textwidth}
\includegraphics[width = \textwidth]{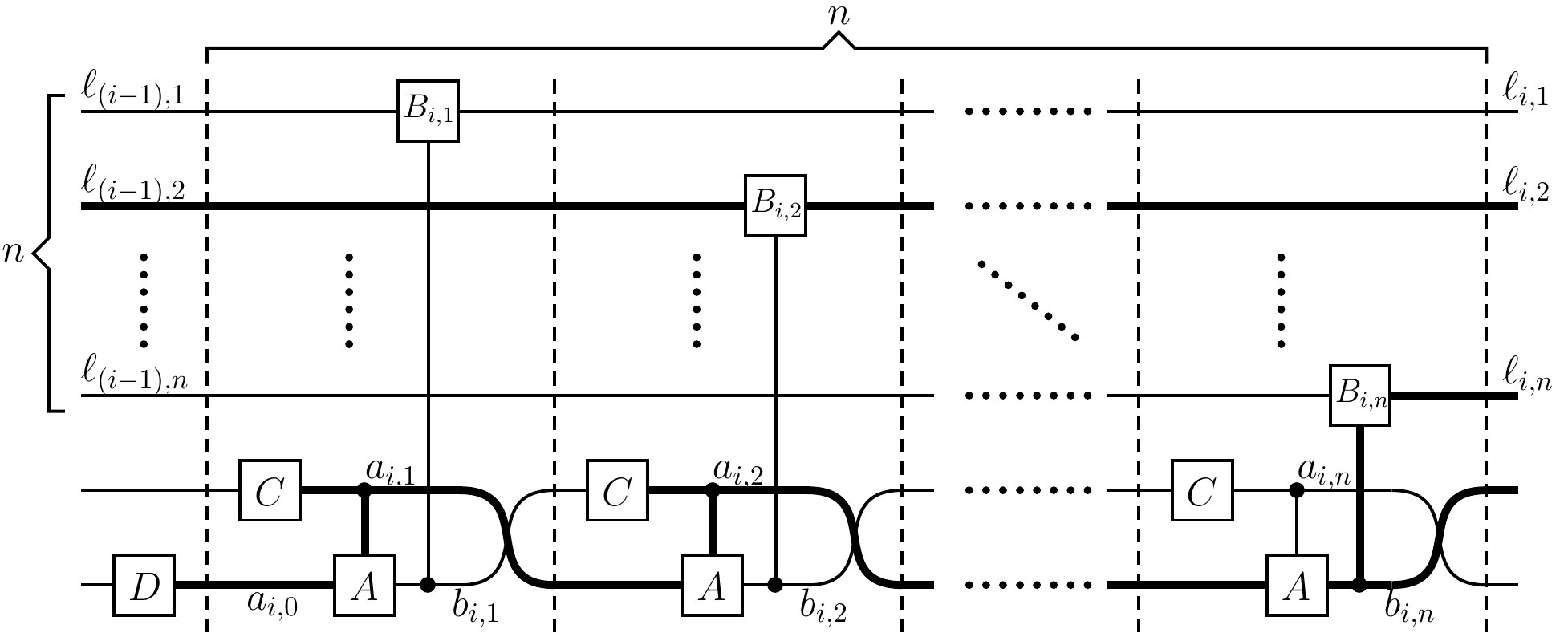}
\end{subfigure} 
\\ \vspace{0.8cm}
\centering
\begin{subfigure}[h!]{0.6\textwidth}
\centering
\includegraphics[width = \textwidth]{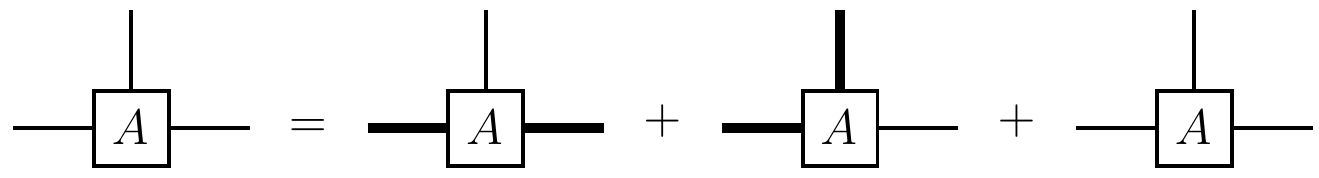}
\end{subfigure}
\\  \vspace{0.5cm}
\begin{subfigure}[h!]{0.6\textwidth}
\centering
\includegraphics[width = \textwidth]{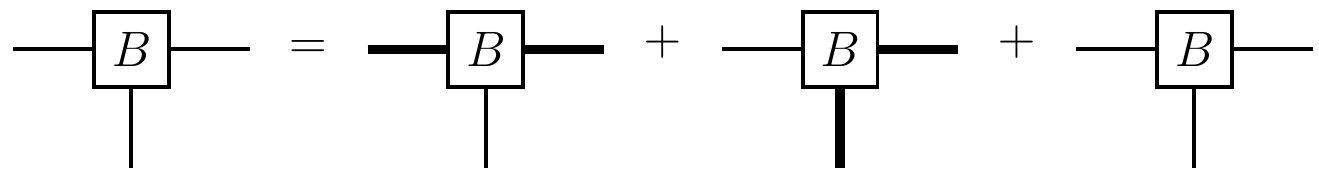}
\end{subfigure}
\\  \vspace{0.5cm}
\begin{subfigure}[h!]{0.6\textwidth}
\centering
\includegraphics[width = \textwidth]{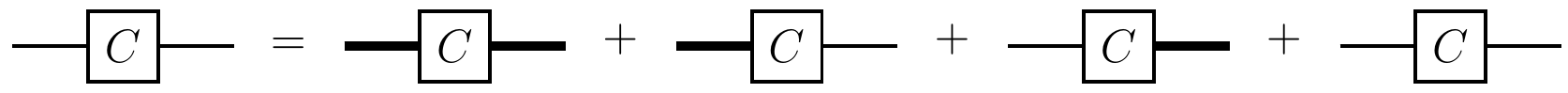}
\end{subfigure}
\\  \vspace{0.5cm}
\begin{subfigure}[h!]{0.4\textwidth}
\centering
\includegraphics[width = \textwidth]{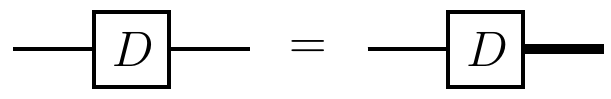}
\end{subfigure}
\\  \vspace{0.5cm}
\begin{subfigure}[h!]{0.4\textwidth}
\centering
\includegraphics[width = \textwidth]{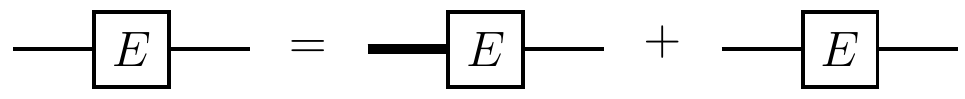}
\end{subfigure}
\\  \vspace{0.5cm}
\begin{subfigure}[h!]{0.4\textwidth}
\centering
\includegraphics[width = \textwidth]{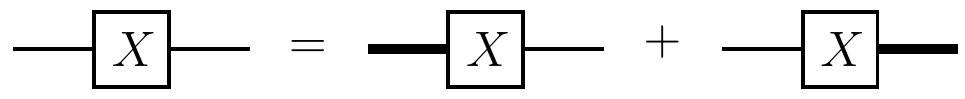}
\end{subfigure}
\caption{A pictorial proof of Theorem~\ref{theo:uncon}. To contract the network in Figure \ref{fig:T_perm}, first fix a labeling of the edges and then multiply together the corresponding tensor elements from each 
tensor.  Then, sum over all such labelings. Note that if any of these tensor elements is zero, then the corresponding term in the sum is also zero. We now illustrate that there is a one-to-one correspondence between 
the nonzero terms in the sum and the terms in the $n$ by $n$ permanent. Namely for every permutation 
$\pi \in S_n$, the product $x_{1,\pi(1)} \ldots x_{n,\pi(n)}$ appears as a nonzero term, and every other term is zero.
In this figure, the wires with thicker lines are labeled $1$, while the thinner lines are labeled $0$. One can check that no other labeling contributes to the sum. To visualize this, we have listed all the nonzero labeling combinations for the $A$, $B$, $C$, $D$, $E$, and $X$ tensors.}
\label{prooffig}
\end{figure}
\end{proof}

%\centering
%\begin{subfigure}[h!]{\textwidth}
%\includegraphics[width = \textwidth]{proof_picture_1.pdf}
%\caption{Tensor network $T_{perm}$.}
%\label{fig:perm_4}
%\end{subfigure} \\

It is worth noting that probably \emph{most} skeletons have large monotone complexity using \emph{any} monotone method, and our reduction to a skeleton for the permanent is just a choice.  We choose the permanent as one of the few candidates for which an unconditional monotone lower bound is known.

Also worth noting is that many hard skeletons may be realized by easy circuits.  In the proof of Theorem~\ref{theo:uncon}, the quantum circuit we construct from Figure~$\ref{fig:C_perm}$ is Clifford, and so it can be simulated in \emph{polynomial} time classically.  Probably, there are quantum circuits which realize the same skeleton and are fundamentally hard to simulate. Monotone methods cannot distinguish between the two: they are oblivious to certain circuit structures altogether.

Interestingly, in the proof of Theorem~\ref{theo:uncon}, contracting the tensor $T_{perm}$ from left to right yields a monotone circuit computing the permanent in time $O(n^2 2^n)$.  To our knowledge, this is the fastest such monotone algorithm.  If restricted to monotone circuits, naive enumeration over all permutations takes time $\Omega(n!)$, whereas a Savitch-type trick (as in~\cite{aaronson2017complexity}) reduces the time complexity to $O(4^n)$.  The fastest general algorithms computing the permanent are non-monotone, and run in roughly $O(n2^n)$ time \cite{brualdi1991combinatorial, glynn2010permanent}. It would be interesting if our method could be modified to yield a monotone circuit matching the lower bound of $n(2^{n-1}-1)$.

Recently, Austrin et al.~\cite{austrin2017tensor} considered proving unconditional lower bounds in a similar context. Starting from a multilinear polynomial, they study the complexity of evaluating this polynomial using tensor network contraction. Roughly speaking, they define the associated time-complexity to be the size of the smallest arithmetic circuit computing the polynomial, subject to the restriction that the circuit may be realized as a tensor network contraction. In our model, we restrict to evaluating monotone arithmetic circuits but with arbitrary structure.  They allow evaluation of non-monotone arithmetic circuits, but restrict to those represented by a tensor network contraction. 

While their goal is to study a new computational model, our goal is to lower bound strong quantum simulators.  Consequently, our models are restrictive in different ways and so our lower bounds are incomparable.

%% file: conditional.tex
\section{Conditional lower bounds for strong simulators}

Theorem \ref{theo:uncon} only holds for monotone simulation methods.
To prove a super-polynomial lower bound on general strong simulation,
we turn to a conditional hardness argument\footnote{Note that an unconditional proof would yield an advance in one of the hardest problems in theoretical computer science, showing $P \neq \#P$; no algorithm can count the number of solutions of Boolean expressions of size $n$ in time $n^{O(1)}$.}.
In this section we prove that $2^{n-o(n)}$ is a lower bound for {\em any} strong simulator with accuracy $2^{-n}$ under the Strong Exponential Time Hypothesis (SETH).

Recall that conditional hardness always takes the form: if problem A is hard then problem B is hard;
although a conditional hardness proof does not deliver absolute evidence for the hardness of B,
it is protected by the hardness of A.  The best choice for A is then a problem for which we have overwhelming evidence of hardness.

One of the most studied problems in computer science is the SAT problem together with its special cases, $k$-SAT for different $k$. 

\medskip

\noindent{\bf The SAT Problem.} The input for SAT with size parameters $n$ and $m$ is a Boolean formula
\[
\phi(x_{1},\ldots,x_{n}) = \bigwedge_{i=1}^{m}\left( \bigvee_{j=1}^{k_{i}} l_{i,j}\right) \;\;\;\;\;\;\;\; l_{i,j}\in \{x_{1},\ldots,x_{n}, \neg x_{1},\ldots,\neg x_{n}\}
\]
where 
$l_{i,j}$ are called \emph{literals} and the sub-expressions 
$C_i=\bigvee_{j=1}^{k_i}l_{ij}$ are called \emph{clauses}. Without loss of generality we require that every variable has at most one occurrence in every clause, implying
$k_{i}\le n$ for $1\le i\le m$. We are interested in formulas $\phi$ with polynomial length: 
$m=n^{O(1)}$. 

Given $\phi$, the task is to determine whether there exists an assignment $x_{i}\rightarrow\{0,1\}$ ($1\le i\le n$) such that $\phi(x_1,\ldots, x_n)=1$,
i.e. if $\phi$ is satisfiable.

\medskip

\noindent{\bf The $k$-SAT Problem.} The $k$-SAT problem is the SAT problem with the restriction that every instance $\phi$ has every $k_{i}$ at most $k$.

\medskip

\noindent The following two questions are unresolved.

\medskip

\noindent\begin{tabular}{lp{4.8in}}
Question 1. & Can the 3-SAT problem be solved in time $(1+\epsilon)^{n} {\rm poly}(m)$ for every $\epsilon > 0$? Algorithms have been repeatedly improved upon \cite{paturi1997satisfiability, paturi2005improved, schoning1999probabilistic, hertli20143} to reach the current state of the art at around $1.3^n$ time. \\[2ex]
Question 2. & Can the SAT problem be solved in time $2^{\alpha n}{\rm poly}(m)$ for some $\alpha<1$, when $m = poly(n)$?
The best current algorithm \cite{calabro2006duality} runs in time ${2^n \over 2^{n/O(\log m/n)}}$.
\end{tabular}

\medskip

An improved algorithm for SAT would make a tremendous impact on many areas of computer science.  This has lead to the somewhat widespread belief that reducing the time-complexity of the SAT problem will hit a hard limit.  This belief has been formalized in the following two commonly held hypotheses.

\medskip

\noindent{\bf Exponential Time Hypothesis (ETH).} The answer to Question 1 is no.  There is an $\epsilon>0$ such that the
time complexity of 3-SAT is at least $(1+\epsilon)^{n} {\rm poly}(m)$.

\medskip

\noindent{\bf Strong Exponential Time Hypothesis (SETH).} The answer to Question 2 is no. Any algorithm deciding SAT must take at least $2^{n-o(n)}{\rm poly}(m)$ time. 

\medskip

We show that if a strong quantum simulator can reach a certain efficiency, then the SETH would be refuted. This efficiency is quantified in terms of both the time-complexity and accuracy of the simulator.

\begin{Theorem}\label{mainconditional}
Assume the SETH holds. Then a universal quantum simulator which can 
approximate \emph{any} output amplitude to precision $2^{-n}/2$ on a quantum circuit with poly$(n)$ operations must take at least $2^{n-o(n)}$ time.
\end{Theorem}
The proof of Theorem \ref{mainconditional}
uses a reduction from SAT to argue that if the simulator can determine an amplitude up to a certain precision, then it could solve a corresponding SAT problem.
The reduction can be summarized as follows.
\begin{itemize}
\item[$(i)$]{For a SAT instance $\phi$ construct a reversible circuit ${\cal C}' = {\cal C}'_{\phi}$ with sub-linear space overhead and polynomial time overhead which can compute $\phi(x)$ for an assignment $x$.}
\item[$(ii)$]{Choose a basis state (e.g. $|0\ldots 0\rangle$) which counts the number of assignments satisfying $\phi$ in its amplitude when running $C_{\phi}$,
a quantum circuit constructed from ${\cal C}'$:
\[
\langle 0 \ldots 0 | {\cal C}_{\phi} |0\ldots 0\rangle = { |\,  \{ x\in \{0,1\}^{n}\;  |\;  \phi(x) = 1\}\, |  \over 2^{n}}
\]}
\end{itemize}
Step $(i)$ is purely classical while step $(ii)$ utilizes the quantum power of the simulator.
Any lower-bound on the time-complexity of SAT
then implies a lower bound on the time complexity of the simulator.  
Theorem \ref{mainconditional} utilizes the strongest conjectured lower bounds available for SAT to imply the sharpest conjectured lower bounds for the parameters of a strong simulator.
To push our bounds even further, we determine the threshold parameters beyond which a simulator would improve upon the best known algorithms solving SAT.

\begin{Theorem}\label{sharpest}
Assume there is a strong simulator that runs in time ${2^{n} \over 2^{n/o(\log m/n)}} $. Then this would improve on the running time of the best SAT solver.
\end{Theorem}

In section~\ref{rereversible} we address point $(i)$.  Next, in section~\ref{reduction}, we address point $(ii)$ and prove Theorems~\ref{mainconditional} and~\ref{sharpest}.
%%%%%%%%%%%%%%%%%%%%%%%%%%%%%%%%%%%%%%%%%%%%%%%%%%%%%%%%%%%%%%%%%%%%%%%%%
%%%%%%%%%%%%%%%%%%%%%%%%%%%%%%%%%%%%%%%%%%%%%%%%%%%%%%%%%%%%%%%%%%%%%%%%%

\subsection{Reversible evaluation of a SAT formula} \label{rereversible}

In this section we construct a reversible circuit ${\cal C}'$ evaluating a given SAT formula using sub-linear space overhead and polynomial time overhead.
This problem was first efficiently solved by Charles H. Bennett in 1989 \cite{bennett1989time}. For the sake of completeness and explicit constants, we reproduce the argument here, but emphasize that our 
construction follows \cite{bennett1989time}. Bennett expresses the problem in the 
language of Turing Machines rather than circuits,\footnote{A uniform rather than non-uniform computational device.} but the proof ideas
are otherwise identical.

%A reversible classical gate $F$ on $n$ bits is simply an invertible function $F:\{0,1\}^n\rightarrow \{0,1\}^n$. One example of a reversible classical gate is the 
%Recall the action of the reversible Toffoli gate $\toffoli(x,y,z)=(x,y,z\oplus (x\land y))$ acting on 3 bits. 
%A classical circuit is called reversible if it consists of only reversible gates. 

\medskip

\noindent{\bf Reversible Circuits.} A reversible classical circuit consists of reversible gates. A reversible classical gate $F$ is simply an invertible function 
$F:\{0,1\}^d\rightarrow \{0,1\}^d$, for some $d$. Typically $d$ is one, two, or three. An important example of a reversible classical gate is the 
Toffoli gate $\toffoli(x,y,z)=(x,y,z\oplus (x\land y))$ acting on 3 bits. 

\medskip

\noindent{\bf Universal Gate Set.} Throughout, our choice of universal gate set for reversible computation is
\[
{\cal G} = \{\toffoli,\cnot,\nnot\}.
\] 

The above choice influences the circuits for which our lower bound argument holds up to constants.  
%in the following stronger sense: 
%`Even if we restrict ourselves to the above types of gates plus Hadamard gates, our lower bound holds.''
The particular gate set we have chosen has the nice property that all its elements are self-inverse.
 %which 
%will make reversing a computation to be described by simply reversing the order of the gates.

\medskip

\noindent{\bf Tidy Computation.} We say a reversible circuit $C:\{0,1\}^{n+a(n)+1}\rightarrow \{0,1\}^{n+a(n)+1}$ \emph{tidily computes} a function $f:\{0,1\}^n\rightarrow \{0,1\}$ if 
\begin{equation}\label{reversible}
\forall \; x\in\{0,1\}^n, y\in\{0,1\}, \;\;\; C(x,0^{a(n)},y)=(x,0^{a(n)},y\oplus f(x)).
\end{equation}

\medskip

We call the $a(n)$ extra bits the \emph{ancilla} bits. The tidiness comes from the fact that the input and the ancilla bits are restored by the end of the computation.
%The \emph{width} $w_C$ of the circuit $C$ is defined as $w_C = n+a(n)+1$, the \emph{size} $s_C$ of the $C$ is defined as the number of gates in the circuit,
%and the \emph{depth} of $C$, $d_{C}$, is the maximum length of the longest (directed) path from any input gate to any output gate. 
%The first $n$ bits are called the input bits, the last one the output bit, 
%Ancilla bits may be necessary in reversible computation either
%because they can be used to store intermediate computation results, 
%which may later be used or remain unused, but cannot be thrown away because of reversibility.
%For a function, the space-time complexity tradeoff for irreversible and reversible computation can be very different because of the presence of ``garbage'' values. Such values would occupy space which can otherwise be used for further computation; in order to ``erase'' these values in a reversible way, the computation needed to produce those values must be undone in reverse order, which takes the same amount of time of producing those values.
%Compared to irreversible circuits, reversible circuits have the interesting property that it does not induce irreversible thermodynamic cost. This means that at least in principle, there is no heat produced in the procedure of the computation, hence reverse computation might be a favorable computation model if computers with extremely low energy cost is preferr
From the perspective of quantum circuits, every reversible gate can be seen as a unitary transformation. %which takes $|i\rangle$ to $|F(i)\rangle$. 
By linearity, the action of a reversible circuit $C$ computing $f$ in Equation (\ref{reversible})
can be extended to an action on the Hilbert space $C^{\{0,1\}^{n}}$:
%\[
%C|x,0^{a(n)},y\rangle = |x,0^{a(n)},y\oplus f(x)\rangle.
%\]
%By linearity,
\medskip
\begin{center}
$C\;\;\;\;$ applied to $\;\;\;\;\;\;\sum_{x}\alpha_x |x,0^{a(n)},0\rangle\;\;\;\;\;\;$ is $\;\;\;\;\;\;\sum_x \alpha_x|x,0^{a(n)},f(x)\rangle$.
\end{center}
\medskip

%In short, every circuit $C$ that computes a function $f$ performs the following unitary:
%\[
%U \; = \; U_{f} \; := \; \sum_{x}|x\rangle\langle x|\otimes X^{f(x)} \;\;\;\;\;\;\;\;\;  {\rm Here} \; X\; \mbox{is the pauli matrix} \;
%\left(\begin{array}{ll}
%0 & 1\\
%1 & 0
%\end{array}\right).
%\]
\noindent Note that the ancilla bits are decoupled from the input and output registers.

\begin{Lemma}[{\bf MAIN}, \cite{bennett1989time}]\label{MAINlemma}
Suppose $\phi$ is a SAT formula with $n$ variables and $m$ clauses.  %For the sake of simplicity we assume that $m$ and $n$ are powers of two.
Then there is a circuit ${\cal C}'$ that tidily computes $\phi$ with
\[
s({\cal C}') \le   8 \times 3^{\lceil \log n \rceil +  \lceil \log m \rceil} -1, \;\;\;\;\;\;\; w({\cal C}') \le n + 2(\lceil \log n \rceil +  \lceil \log m \rceil ).
\]
\end{Lemma}

%\noindent{\bf Step $(i)$} Building an efficient reversible circuit for evaluating a SAT formula $\phi$.  

%\begin{figure}[htb!]
\begin{wrapfigure}{r}{4.5cm}
\begin{tabular}{lp{3.3cm}}
$\;$ & 
%\begin{center}
%\begin{tabular{ll} 
\includegraphics[width=4.5cm]{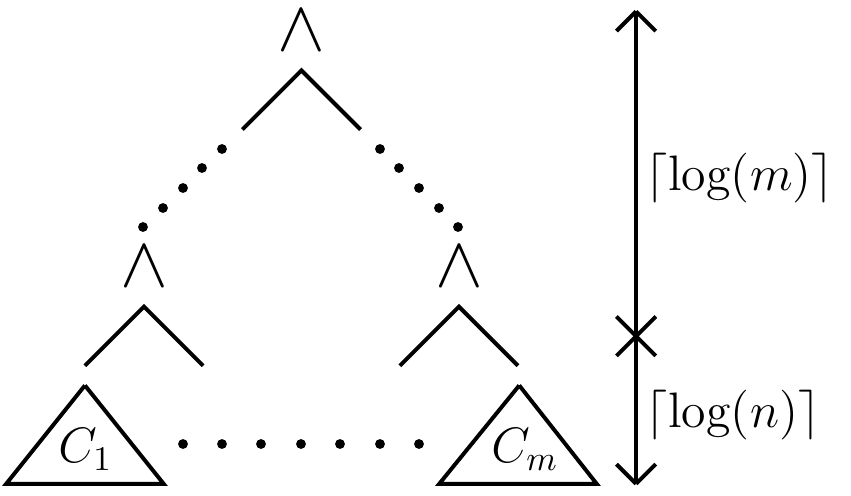}  
%\end{tabular}
%\end{center}
\caption{\label{ANDOR} SAT instance $\phi$ expressed as a binary tree.} 
\vspace{-0.1in}
\end{tabular}
\end{wrapfigure}
%\end{figure}
\noindent Our proof scheme follows that in \cite{bennett1989time}. Notice that $\phi$ can be written as a binary tree with ANDs and ORs as in Figure \ref{ANDOR}.
Thus, we will first implement ANDs and ORs of circuits in a reversible way.
We will then recursively compose ${\cal C}'= {\cal C}'_{\phi}$ by exploiting the binary tree structure. In the rest of the section, we will prove Lemma~\ref{MAINlemma} using these steps. To get the best constants we first define the following.

\medskip

\noindent{\bf Untidy Computation.} We say that a reversible circuit $C:\{0,1\}^{n+a(n)+1}\rightarrow \{0,1\}^{n+a(n)+1}$ \emph{untidily computes} a 
function $f:\{0,1\}^n\rightarrow \{0,1\}$ if $C(x,0,0) = (\ast,\ast, f(x))$ for all $x\in\{0,1\}^n$.
%(we drop the assumption from Equation (\ref{reversible}) that the input and ancilla wires in $C(x,0,0)$ are restored,
%and we drop any assumption whatsoever about $C(x,0,1)$). 

\begin{Lemma}\label{weaktostrong}
We can convert a circuit $U$ that untidily computes $f$ into a circuit that tidily computes $f$. The conversion
doubles the size of the circuit along with adding one extra CNOT gate and one extra ancilla wire, as shown in Figure \ref{weak}.
\end{Lemma}

\begin{proof}
The circuit in Figure \ref{weak} executes $UGU^{-1}$, where $G$ is a CNOT gate acting on the 
new output wire ($b$ in the figure), and controlled by the old output ($a$ in the figure). More precisely,
\begin{eqnarray*}
|x,0,\underbrace{0}_{a},\underbrace{0}_{b}\rangle  & \stackrel{U}{\longrightarrow} & |\alpha,\beta,f(x),0\rangle \; \stackrel{G}{\longrightarrow} \;
|\alpha,\beta,f(x),f(x)\rangle 
\; \stackrel{U^{-1}}{\longrightarrow} \;
|x,0,0,f(x)\rangle \\
|x,0,\underbrace{0}_{a},\underbrace{1}_{b}\rangle & \stackrel{U}{\longrightarrow} & |\alpha,\beta,f(x),1\rangle \; \stackrel{G}{\longrightarrow} \;
|\alpha,\beta,f(x),\neg f(x)\rangle 
\; \stackrel{U^{-1}}{\longrightarrow} \;
|x,0,0,\neg f(x)\rangle
\end{eqnarray*}
Observe that the added wire carries the output, while the original output wire serves as an ancilla.
\end{proof}

\begin{shaded}
{\small
\noindent{\bf History.} The following is how Bennett describes the principle of cleaning up unwanted data. % in a slightly different context: 

%\noindent{\bf Remark.}
%Bennett describes the principle of cleaning up unwanted data (the principle we employ above) in his paper in a slightly different context: 
%\medskip
``{\em The simulation {\rm [of a non-reversible Turing Machine]} proceeds 
in three stages. The first stage uses an extra tape, initially blank, to record all the information that would have been thrown away by the irreversible computation being simulated.
This history-taking renders the computation reversible, but also leaves a large amount of unwanted data on the history tape (hence the epithet ``untidy''). The second stage copies the 
output produced by the first stage onto a separate output tape, an action that is intrinsically reversible (without history taking) if the output tape is initially blank. 
The third stage exactly reverses the work of the first stage, thereby restoring the whole machine, except for the now-written output tape, to its original condition.
This cleanup is possible exactly because the first stage was reversible and deterministic.}''}
\end{shaded}
%``{\em The simulation {\rm [of a non-reversible Turing Machine]} proceeds 
%in three stages. The first stage uses an extra tape, initially blank, to record all the information that would have been thrown away by the irreversible computation being simulated.
%This history-taking renders the computation reversible, but also leaves a large amount of unwanted data on the history tape (hence the epithet ``untidy''). The second stage copies the 
%output produced by the first stage onto a separate output tape, an action that is intrinsically reversible (without history taking) if the output tape is initially blank. 
%The third stage exactly reverses the work of the first stage, thereby restoring the whole machine, except for the now-written output tape, to its original condition.
%This cleanup is possible exactly because the first stage was reversible and deterministic.}''

\medskip

\begin{figure}
\begin{center}
\includegraphics[width=4.5cm]{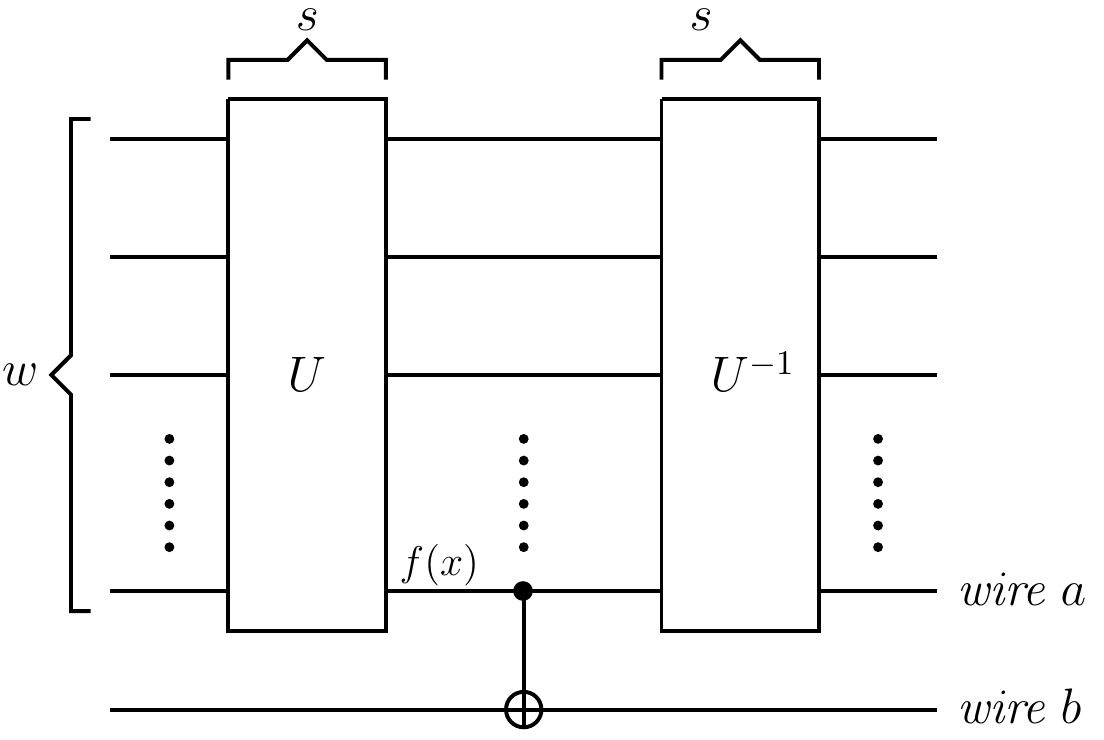}  
\end{center}
\caption{\label{weak} The construction that restores the input and ancilla wires after an untidy computation of $f(x)$ and produces the output on the added wire $b$. 
Initially, the output of $U$ was sent to wire $a$.} 
\end{figure}

The main building blocks of ${\cal C}'$ that implement binary AND and OR as reversible circuits come from the following lemma.

\begin{Lemma}
\label{lem:block}
Suppose $U_1$ and $U_2$ untidily compute $f_1$ and $f_2$ with width at most $w$ and sizes $s_1$ and $s_{2}$, respectively.
Then there is a circuit
of width at most $w+2$ and size $2s_{1}+ s_{2}+ 2$ that untidily computes $f_{1} \wedge f_2$, and a circuit
of width at most $w+2$ and size $2s_{1}+ s_{2}+ 5$ that untidily computes
$f_{1} \vee f_{2}$.
\end{Lemma}

Note that in Lemma \ref{lem:block}, we could have performed tidy computations rather than untidy computations with meager overhead. We avoid this for two reasons.  First, the untidy computation has a
simpler circuit. Second, the constant factor loss will
culminate in a polynomial loss due to the iterative construction in the proof of Theorem \ref{mainconditional}. Thus, it is more economical to untidily compute everything until the end,
and then make the circuit tidy by employing Lemma \ref{weaktostrong}. The bounds are specific to the gate set ${\cal G}$, 
which results in the slight asymmetry between the AND and OR circuits.

\begin{proof}[Proof of Lemma~\ref{lem:block}]
The reversible AND and OR of the circuits, with the appropriate sizes and widths, are constructed explicitly below.

\medskip

\begin{center}
\begin{tabular}{ccc}
\includegraphics[width=5cm]{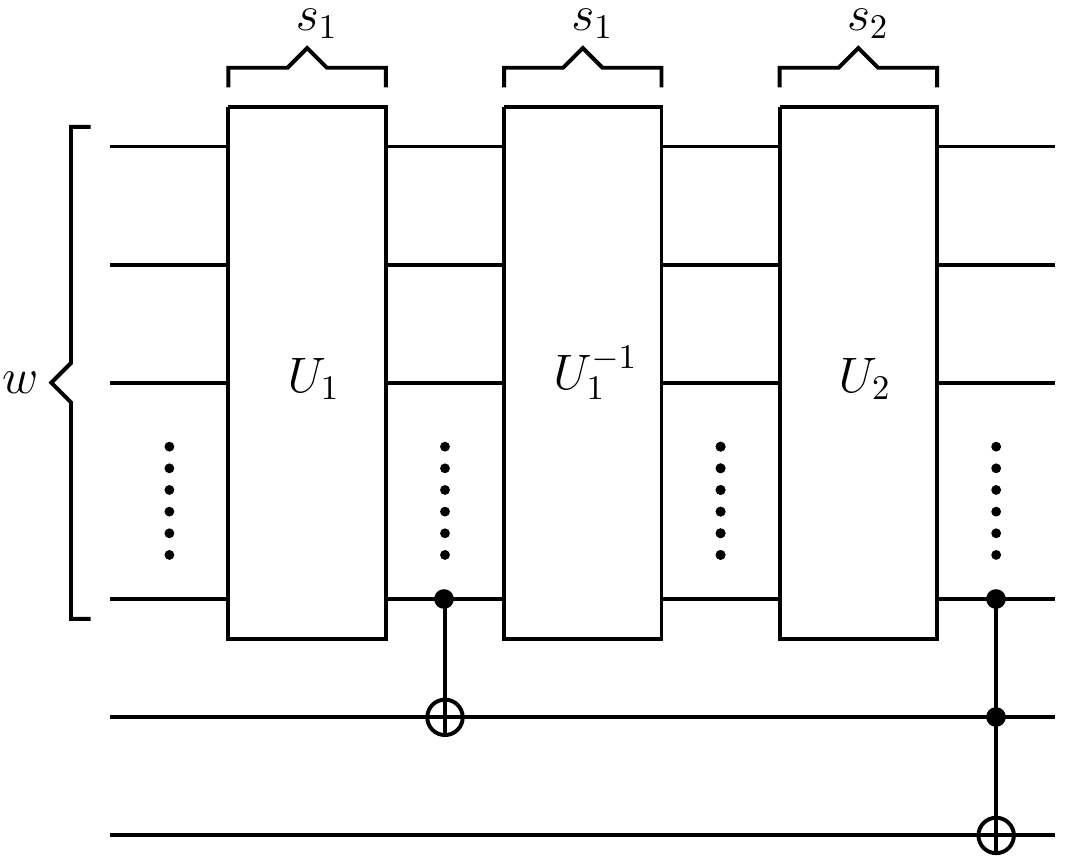}   & $\hspace{0.1in}$ & \includegraphics[width=5.7cm]{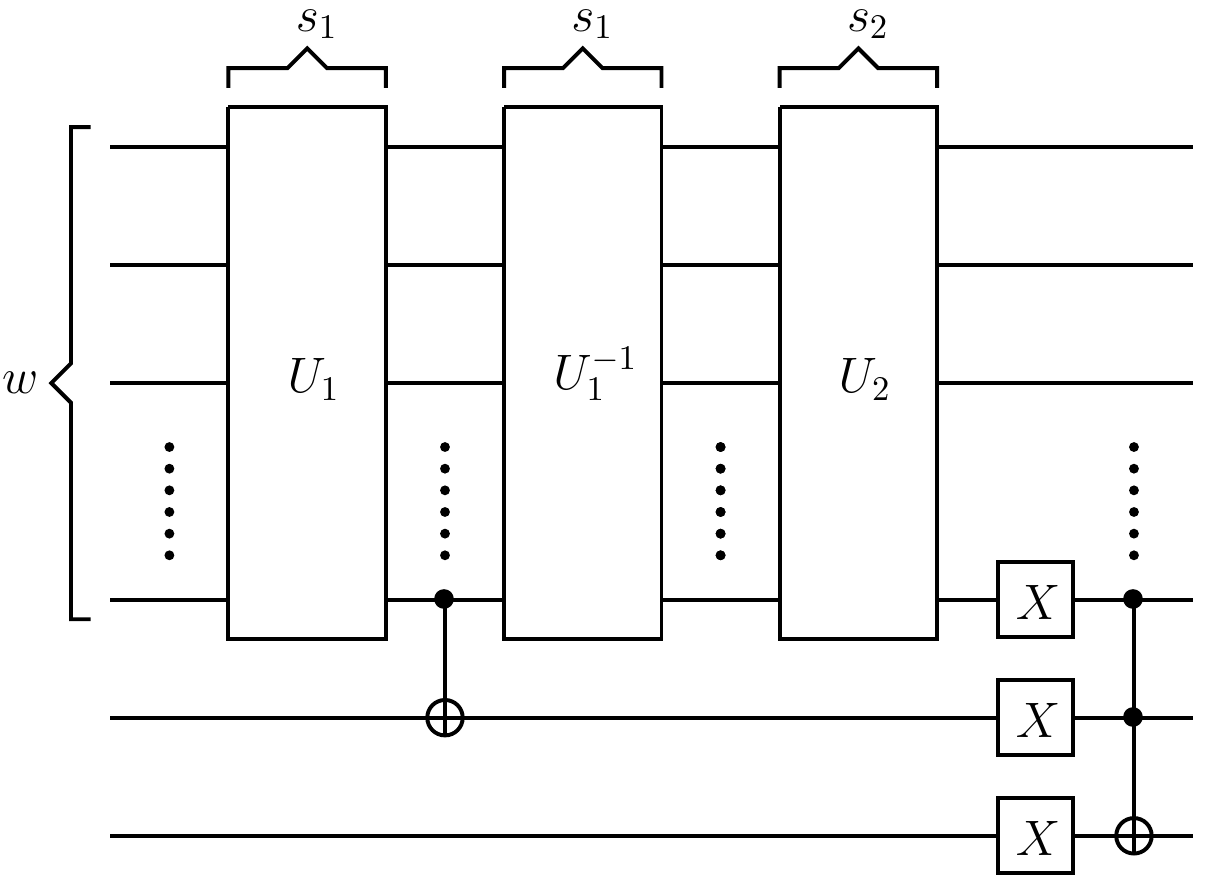}   \\[5pt]
Circuit for $f_{1}\wedge f_{2}$    &     &      Circuit for $f_{1}\vee f_{2}$ 
\end{tabular}
\end{center}
\end{proof}

\begin{Corollary} \label{largenum}
Untidily computing $f_{1}\wedge \ldots \wedge f_{n}$ takes width at most $w+ 2\lceil \log n \rceil$ and size at most $ 3^{\lceil \log_{2} n \rceil} (s+1)-1$,
where $s$ and $w$ are the simultaneous size and width upper bounds for circuits that untidily-compute the set of $f_{i}$. 
\end{Corollary}
\begin{Corollary} \label{largenum2}
untidily computing $f_{1}\vee \ldots \vee f_{n}$ takes width at most $w+ 2\lceil \log n \rceil$ and size at most $3^{\lceil \log_{2} n \rceil}(s+\frac52)-\frac52$.
\end{Corollary}

\noindent{\bf The Reversible Circuit for $\phi = C_{1}\wedge \ldots\wedge C_{m}$.} We first create circuits ${\cal C}_{1},\ldots,{\cal C}_{m}$ that untidily compute clauses $C_{1},\ldots,C_{m}$,
using Corollary \ref{largenum2}. 
From these circuits, we then construct a reversible circuit that 
untidily computes their conjunction, using Corollary \ref{largenum}. 
The resulting circuit will untidily compute $\phi$ with the following parameters.

\begin{Corollary}\label{smallcor} Suppose $\phi$ has $n$ variables and $m$ clauses.  Then we can untidily compute $\phi$ using a reversible circuit of width at most $w+ 2(\lceil \log n \rceil + \lceil \log m \rceil ) $ and size 
at most $ 4\times 3^{\lceil \log_{2} n \rceil  + \lceil \log_{2} m \rceil} -1$.
\end{Corollary}

Finally, the centerpiece Lemma~\ref{MAINlemma} then follows from applying Lemma \ref{weaktostrong} to the untidy circuit in Corollary~\ref{smallcor}.
\clearpage
\subsection{Reducing SAT to strong simulation} \label{reduction}

\begin{wrapfigure}{r}{4.5cm}
\vspace{-0.35in}
\begin{tabular}{lp{3.3cm}}
$\;$ & \includegraphics[width=3cm]{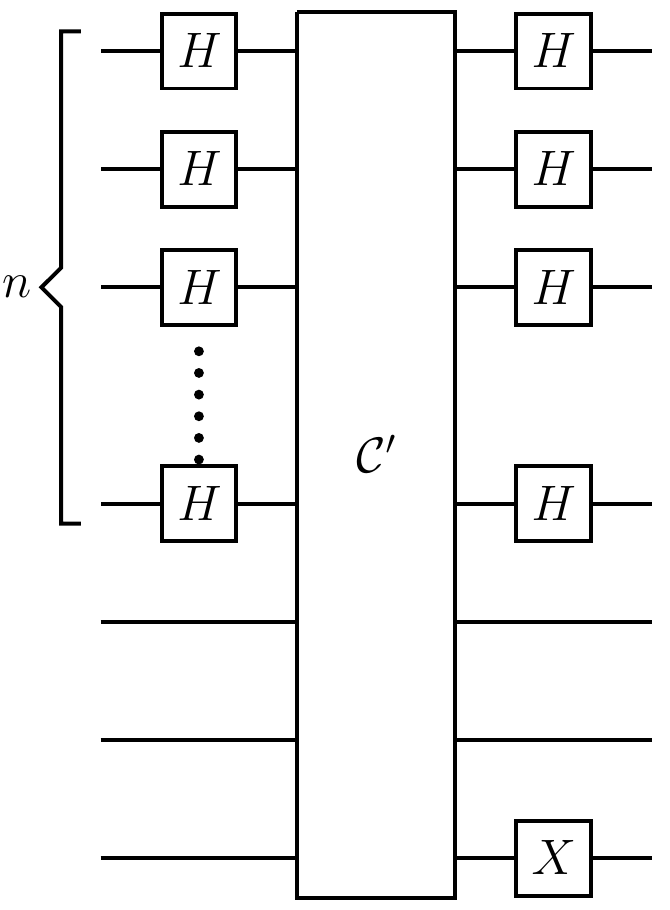}  \\
 & The construction of ${\cal C}_{\phi}$, where ${\cal C}'$ comes from 
 Section \ref{reversible}.
 %If the output (bottom wire) was not negated,
%$ \langle 0 \ldots 0 | {\cal C}_{\phi} |0\ldots 01\rangle$ would have to
%replace $ \langle 0 \ldots 0 | {\cal C}_{\phi} |0\ldots 00\rangle$. %in Equation (\ref{sattocircuit}).
\end{tabular}
\end{wrapfigure}

Given a SAT instance $\phi$ with $n$ variables and $m$ clauses, we would like to construct a quantum circuit ${\cal C}_{\phi}$ so that 

\begin{equation}\label{sattocircuit}
\vartheta := \langle 0 \ldots 0 | {\cal C}_{\phi} |0\ldots 0\rangle = { |\,  \{ x\in \{0,1\}^{n}\;  |\;  \phi(x) = 1\}\, |  \over 2^{n}}.
\end{equation}

Let ${\cal C}'_{\phi}$ be a (classical) reversible circuit coming from Lemma \ref{MAINlemma} that tidily computes $\phi$.  Then the quantum circuit ${\cal C}_{\phi}$ %in Figure \ref{wholecirc}
on the right satisfies Equation \ref{sattocircuit}.

%\caption{\label{wholecirc} The final construction using circuit ${\cal C}'$ of Section \ref{reversible}}  

\subsection{Relating the parameters}\label{parameters}
%Let ${\cal C}'_{\phi}$ be the classical reversible circuit constructed in Lemma \ref{MAINlemma} that tidily computes $\phi$.  Then the quantum circuit ${\cal C}_{\phi}$ in Figure \ref{}
%satisfies Equation \ref{sattocircuit}. 
By comparing the parameters of $\phi$ and ${\cal C}_{\phi}$, we can tie the
complexity of the approximate strong simulation problem to the complexity of the SAT problem.

\begin{proof}[Proof of Theorem~\ref{mainconditional}]
Suppose we had a strong simulator that could approximate $\vartheta$ to
within an additive error of $2^{-n}/2$. Then running the simulator on ${\cal C}_{\phi}$, we would be able to tell if $\phi$ is satisfiable or not: if $\vartheta < 2^{-n}/2$ then $\phi$ is not satisfiable, otherwise it is.

Let the number of qubits of ${\cal C}_{\phi}$ be $n' = w({\cal C}_{\phi})$ and its size be $s = s({\cal C}_{\phi})$.
Suppose we could run this simulator in time $2^{(1-c)n'} s^{O(1)}$ for some $c>0$. Then,
substituting the bounds for $w({\cal C}_{\phi})$ and $s({\cal C}_{\phi})$ in Lemma~\ref{MAINlemma},
we obtain a running time of
\[
2^{(1-c)(n + 2\log n + 2\log m)} (nm)^{\log_{2} 3} = 2^{(1-c)n} (mn)^{O(1)}.
\]
This would contradict the SETH.
\end{proof}

\begin{proof}[Proof of Theorem~\ref{sharpest}]
  Suppose that our simulator can approximate amplitude $\vartheta$ to within accuracy $2^{-n}/2$ and in time
${2^{n} \over 2^{n/o(\log m/n)}} $. Then this simulator would solve the SAT problem in time
\[
{2^{n + 2\log n + 2\log m}\over 2^{n/o(\log m/n)  } }
\]
which is better than  ${2^n \over 2^{n/O(\log m/n)}}$, beating the currently known best SAT solver \cite{calabro2006duality}.
\end{proof}

\noindent Note that a simple padding argument further extends Theorem~\ref{mainconditional} to the following.

\begin{Corollary}
Assume the SETH and let $0<\alpha \leq 1$.  Then any strong simulator with approximation precision $2^{-\alpha n}/2$ must take $2^{\alpha n - o(n)}$ time.
\end{Corollary}

%% file: conclusion.tex
\section{Conclusion} 

\subsection{Summary}
In this paper, we proved explicit lower bounds for existing strong simulation methods.  Almost all prominent simulation methods fall into the class of monotone methods, which only focus on the locations of non-zero entries (i.e.\ the \emph{skeleton} of the tensor network) and are oblivious and monotone with respect to the nonzero coefficients. We then reduce the problem of computing the permanent using monotone arithemetic circuits to evaluating a quantum circuit using monotone methods. Therefore, the hardness of the former implies the hardness of the latter.

For general strong simulators, we showed that they can be harnessed to solve the \#SAT problem if the simulation method is of a certain accuracy. The \#SAT problem is not only \#P-hard; it is widely believed that solving \#SAT takes about $2^n$ time under the Strong Exponential Time Hypothesis. This allows us to prove an \emph{explicit} conditional lower bound on all strong simulation methods with high accuracy.

Our lower bounds focus on general quantum circuits, showing that there exist hard instances with reasonable size that cannot be evaluated efficiently using current simulation methods. However, such reductions might not be applicable if the circuit is drawn from a restricted class (for example, circuits consisting of only Clifford gates). Proving that a certain restricted class of quantum circuits admits a hard instance requires additional effort. We also suspect that most quantum circuits are hard to simulate in a stronger sense, but we leave this consideration to future work.

%% file: openproblems.tex
\subsection{Open problems}

\noindent{\bf Hardness of Random Instances.} Both of our results focus on the worst-case complexity of strong simulation by artificially constructing hard instances that cannot be simulated efficiently. For example, we reduce to the permanent to apply a known unconditional lower-bound to a particular closed skeleton. 
While we can probably show that \emph{many} closed skeletons are hard by
manipulating the proof of \cite{jerrum1982some}, proving average-case hardness might require a worst- to average-case reduction, similar to~\cite{bouland2018quantum}.

\medskip

\noindent{\bf Space-Efficient Algorithms.} In this paper, we have considered only the time-complexity of a simulator.  More generally, one may be interested in restricting to space-efficient simulators. Along this line, there is the Feynman path-integral which has time complexity $O(2^{nd})$.  More recently, Aaronson and Chen~\cite{aaronson2017complexity} use Savitch's Theorem to show that one can achieve an $O(d^{n})$ time-complexity. Does there then exist an even faster space-efficient strong simulator with time-complexity $O(d \cdot 2^{n})$?  Can we fine-tune such a simulator to achieve a general time-space tradeoff, as in~\cite{aaronson2017complexity}? 

\medskip

\noindent{\bf Larger Additive Gaps.} Theorem~\ref{mainconditional} addresses the hardness of strong simulation up to accuracy $O(2^{-n})$, but it could be that strong simulation up to accuracy $O(2^{-{n/2}})$ would also take $2^{n-o(n)}$ time. Approximation up to accuracy $2^{-{n/2}}$ is particularly interesting because $2^{-{n/2}}$ is the typical value of an amplitude over a randomly chosen basis vector.
A more general question is determining, for some $a < b$, the complexity of deciding whether an amplitude is at most $a$ or at least $b$. 
An avenue for proving such results may come from the hardness of approximating the \#SAT problem. 

\medskip

\noindent{\bf Lower Bounds on Weak Simulators.} Can we prove explicit lower bounds for weak simulators? It is well known that there can be no efficient weak simulator unless $PH$ collapses to the third level. However, it is plausible that some weak simulators may run in sub-exponential (e.g. $2^{\sqrt{n}}$) time, which may be affordable in the near future.

\medskip

\noindent{\bf Superior Weak Simulators.} Given the compelling evidence that strong simulation is hard, can we design a weak simulator that runs general quantum circuits in time $o(d \cdot 2^{n})$?  The simulator constructed by Bravyi and Gosset~\cite{bravyi2016improved} is an excellent example of the potential gains afforded by weak (versus strong) simulation.

%Note that the location of the gap matters. 
%For instance, it might take $2^{n-o(n)}$ time to decide 
%if a SAT instance has at least $2^{n-1}+2^{{n\over 2}-1}$ solutions or it has at most $2^{n-1}-2^{{n\over 2}-1}$ solutions,
%while if the same sized gap is located near zero, the approximation takes time at most $O(2^{-{n\over 2}})$ .

\section{Acknowledgments}

The authors would like to thank Jianxin Chen, Yaoyun Shi, Fang Zhang, and Avi Widgerson for helpful discussions.  In particular, we thank Yaoyun for first posing the question of space-efficient simulation with time-complexity $O(d \cdot 2^{n})$, out of which these ideas grew.  This work was supported by the Alibaba group.

%Explicit hardness results encourage people to look into weak simulators. Little is known about weak simulators yet apart from restricted class of circuits. Either of the following problems are of great interest:
%\begin{enumerate}
%	\item Can we develop an efficient weak simulator for general quantum circuits? The naive weak simulator uses the strong simulator computing the full vector, and runs in time exponential in the number of qubits. Restricting to the universal gate set Clifford$+T$, Bravyi et al.~\cite{} developed a weak simulator running in time roughly $2^{0.223t}$, where $t$ is the number of $T$ gates in the circuit. It would be great if ideas on weak simulators, which do not directly come from strong simulators, can be developed to push the limit on classical simulation further.

%% file: main.bbl
\newcommand{\etalchar}[1]{$^{#1}$}
\begin{thebibliography}{PGN{\etalchar{+}}17}

\bibitem[AA11]{aaronson2011computational}
Scott Aaronson and Alex Arkhipov.
\newblock {The computational complexity of linear optics}.
\newblock In {\em {Proceedings of the forty-third annual ACM symposium on
  Theory of computing}}, pages 333--342. {ACM}, 2011.

\bibitem[AC17]{aaronson2017complexity}
Scott Aaronson and Lijie Chen.
\newblock {Complexity-theoretic foundations of quantum supremacy experiments}.
\newblock In {\em {Proceedings of the 32nd Computational Complexity
  Conference}}, page~22. {Schloss Dagstuhl--Leibniz-Zentrum fuer Informatik},
  2017.

\bibitem[AG04]{aaronson2004improved}
Scott Aaronson and Daniel Gottesman.
\newblock {Improved simulation of stabilizer circuits}.
\newblock {\em {Physical Review A}}, 70(5):052328, 2004.

\bibitem[AKK17]{austrin2017tensor}
Per Austrin, Petteri Kaski, and Kaie Kubjas.
\newblock {Tensor network complexity of multilinear maps}.
\newblock {\em {arXiv preprint arXiv:1712.09630}}, 2017.

\bibitem[Ben89]{bennett1989time}
Charles~H Bennett.
\newblock {Time/space trade-offs for reversible computation}.
\newblock {\em {SIAM Journal on Computing}}, 18(4):766--776, 1989.

\bibitem[BFNV18]{bouland2018quantum}
Adam Bouland, Bill Fefferman, Chinmay Nirkhe, and Umesh Vazirani.
\newblock {Quantum Supremacy and the Complexity of Random Circuit Sampling}.
\newblock {\em {arXiv preprint arXiv:1803.04402}}, 2018.

\bibitem[BG16]{bravyi2016improved}
Sergey Bravyi and David Gosset.
\newblock {Improved classical simulation of quantum circuits dominated by
  Clifford gates}.
\newblock {\em {Physical Review Letters}}, 116(25):250501, 2016.

\bibitem[BIS{\etalchar{+}}18]{boixo2016characterizing}
Sergio Boixo, Sergei~V Isakov, Vadim~N Smelyanskiy, Ryan Babbush, Nan Ding,
  Zhang Jiang, John~M Martinis, and Hartmut Neven.
\newblock {Characterizing quantum supremacy in near-term devices}.
\newblock {\em {Nature Physics}}, 2018.

\bibitem[BISN17]{boixo2017simulation}
Sergio Boixo, Sergei~V Isakov, Vadim~N Smelyanskiy, and Hartmut Neven.
\newblock {Simulation of low-depth quantum circuits as complex undirected
  graphical models}.
\newblock {\em arXiv preprint arXiv:1712.05384}, 2017.

\bibitem[BR91]{brualdi1991combinatorial}
Richard~A Brualdi and Herbert~John Ryser.
\newblock {\em {Combinatorial matrix theory}}, volume~39.
\newblock {Springer}, 1991.

\bibitem[BSS16]{bravyi2016trading}
Sergey Bravyi, Graeme Smith, and John~A Smolin.
\newblock {Trading classical and quantum computational resources}.
\newblock {\em {Physical Review X}}, 6(2):021043, 2016.

\bibitem[CIP06]{calabro2006duality}
Chris Calabro, Russell Impagliazzo, and Ramamohan Paturi.
\newblock {A duality between clause width and clause density for SAT}.
\newblock In {\em {Computational Complexity, 2006. CCC 2006. Twenty-First
  Annual IEEE Conference on}}, pages 7--pp. {IEEE}, 2006.

\bibitem[CZC{\etalchar{+}}18]{chen2018classical}
Jianxin Chen, Fang Zhang, Mingcheng Chen, Cupjin Huang, Michael Newman, and
  Yaoyun Shi.
\newblock {Classical Simulation of Intermediate-size Quantum Circuits: No
  Imminent Quantum Supremacy from Random Circuit Sampling}.
\newblock {\em {In preparation}}, 2018.

\bibitem[CZX{\etalchar{+}}18]{chen2018}
Zhaoyun Chen, Qi~Zhou, Cheng Xue, Xia Yang, Guangcan Guo, and Guoping Guo.
\newblock {64-Qubit Quantum Circuit Simulation}.
\newblock {\em {arXiv preprint arXiv:1802.06952}}, 2018.

\bibitem[FHS65]{feynman1965quantum}
Richard~P Feynman, Albert~R Hibbs, and Daniel~F Styer.
\newblock {\em {Quantum mechanics and path integrals}}.
\newblock {Courier Corporation}, 1965.

\bibitem[Gly10]{glynn2010permanent}
David~G Glynn.
\newblock {The permanent of a square matrix}.
\newblock {\em {European Journal of Combinatorics}}, 31(7):1887--1891, 2010.

\bibitem[Got98]{gottesman1998heisenberg}
Daniel Gottesman.
\newblock {The Heisenberg representation of quantum computers}.
\newblock {\em {arXiv preprint quant-ph/9807006}}, 1998.

\bibitem[Her14]{hertli20143}
Timon Hertli.
\newblock {3-SAT Faster and Simpler---Unique-SAT Bounds for PPSZ Hold in
  General}.
\newblock {\em {SIAM Journal on Computing}}, 43(2):718--729, 2014.

\bibitem[HS17]{haner20170}
Thomas H{\"a}ner and Damian~S Steiger.
\newblock {0.5 petabyte simulation of a 45-qubit quantum circuit}.
\newblock In {\em {Proceedings of the International Conference for High
  Performance Computing, Networking, Storage and Analysis}}, page~33. {ACM},
  2017.

\bibitem[JS82]{jerrum1982some}
Mark Jerrum and Marc Snir.
\newblock {Some exact complexity results for straight-line computations over
  semirings}.
\newblock {\em {Journal of the ACM (JACM)}}, 29(3):874--897, 1982.

\bibitem[JW09]{ji2009non}
Zhengfeng Ji and Xiaodi Wu.
\newblock {Non-identity check remains QMA-complete for short circuits}.
\newblock {\em {arXiv preprint arXiv:0906.5416}}, 2009.

\bibitem[Mon17]{montanaro2017quantum}
Ashley Montanaro.
\newblock {Quantum circuits and low-degree polynomials over}.
\newblock {\em {Journal of Physics A: Mathematical and Theoretical}},
  50(8):084002, 2017.

\bibitem[MS08]{markov2008simulating}
Igor~L Markov and Yaoyun Shi.
\newblock {Simulating quantum computation by contracting tensor networks}.
\newblock {\em {SIAM Journal on Computing}}, 38(3):963--981, 2008.

\bibitem[Nes08]{nest2008classical}
M~Nest.
\newblock {Classical simulation of quantum computation, the Gottesman-Knill
  theorem, and slightly beyond}.
\newblock {\em {arXiv preprint arXiv:0811.0898}}, 2008.

\bibitem[PGN{\etalchar{+}}17]{pednault2017breaking}
Edwin Pednault, John~A Gunnels, Giacomo Nannicini, Lior Horesh, Thomas
  Magerlein, Edgar Solomonik, and Robert Wisnieff.
\newblock {Breaking the 49-qubit barrier in the simulation of quantum
  circuits}.
\newblock {\em {arXiv preprint arXiv:1710.05867}}, 2017.

\bibitem[PPSZ05]{paturi2005improved}
Ramamohan Paturi, Pavel Pudl{\'a}k, Michael~E Saks, and Francis Zane.
\newblock {An improved exponential-time algorithm for k-SAT}.
\newblock {\em {Journal of the ACM (JACM)}}, 52(3):337--364, 2005.

\bibitem[PPZ97]{paturi1997satisfiability}
Ramamohan Paturi, Pavel Pudl{\'a}k, and Francis Zane.
\newblock {Satisfiability coding lemma}.
\newblock In {\em {Foundations of Computer Science, 1997. Proceedings., 38th
  Annual Symposium}}, pages 566--574. IEEE, 1997.

\bibitem[Pre12]{preskill2012quantum}
John Preskill.
\newblock {Quantum computing and the entanglement frontier}.
\newblock {\em {arXiv preprint arXiv:1203.5813}}, 2012.

\bibitem[Pre18]{preskill2018quantum}
John Preskill.
\newblock {Quantum Computing in the {NISQ} era and beyond}.
\newblock {\em {arXiv preprint arXiv:1801.00862}}, 2018.

\bibitem[Rud09]{rudolph2009simple}
Terry Rudolph.
\newblock {Simple encoding of a quantum circuit amplitude as a matrix
  permanent}.
\newblock {\em {Physical Review A}}, 80(5):054302, 2009.

\bibitem[Sch99]{schoning1999probabilistic}
T~Schoning.
\newblock {A probabilistic algorithm for k-SAT and constraint satisfaction
  problems}.
\newblock In {\em {Foundations of Computer Science, 1999. 40th Annual
  Symposium}}, pages 410--414. {IEEE}, 1999.

\bibitem[Str69]{strassen1969gaussian}
Volker Strassen.
\newblock {Gaussian elimination is not optimal}.
\newblock {\em {Numerische Mathematik}}, 13(4):354--356, 1969.

\bibitem[Val02]{valiant2002quantum}
Leslie~G Valiant.
\newblock {Quantum circuits that can be simulated classically in polynomial
  time}.
\newblock {\em {SIAM Journal on Computing}}, 31(4):1229--1254, 2002.

\end{thebibliography}
